\documentclass[pdflatex,sn-mathphys-num]{sn-jnl}
\usepackage{graphicx}%
\usepackage{multirow}%
\usepackage{amsmath,amssymb,amsfonts}%
\usepackage{dsfont}
\usepackage{bbold}
\usepackage{amsthm}%
\usepackage{mathrsfs}%
\usepackage[title]{appendix}%
\usepackage{xcolor}%
\usepackage{textcomp}%
\usepackage{manyfoot}%
\usepackage{booktabs}%
\usepackage{algorithm}%
\usepackage{algorithmicx}%
\usepackage{algpseudocode}%
\usepackage{listings}%
\usepackage{centernot}

\theoremstyle{thmstyleone}
\newtheorem{theorem}{Theorem}
\newtheorem{proposition}[theorem]{Proposition}
\newtheorem{lemma}[theorem]{Lemma}

\newtheorem{corollary}[theorem]{Corollary}
\newtheorem{assumption}[theorem]{Assumption}
\theoremstyle{thmstyletwo}
\newtheorem*{example}{Example}

\theoremstyle{thmstylethree}
\newtheorem{definition}[theorem]{Definition}

\newcommand{\integ}[4]{\int_{#1}^{#2}#3\text{d}#4}
\newcommand{\graph}{\mathcal{G}}
\newcommand{\tder}{\frac{\text{d}}{\text{d}t}}
\newcommand{\diag}[1]{\text{diag}\big(#1\big)}

\newcommand{\new}[1]{{\leavevmode\color{black}#1}}
\newcommand{\newnew}[1]{{\leavevmode\color{black}#1}}

\begin{document}

\title[Graph-theoretic framework for semi-Markov networks: occupation probabilities, hitting times and weak ergodicity breaking]{Graph-theoretic framework for semi-Markov networks: transient and stationary occupation probabilities, hitting times and weak ergodicity breaking}

\author*{\fnm{Julian B.} \sur{Voits} }\email{voits@thphys.uni-heidelberg.de}

\affil{Institute for Theoretical Physics, Heidelberg University, Germany}

\affil{BioQuant-Center for Quantitative Biology, Heidelberg University, Germany}

\abstract{\new{We show that} finite semi-Markov networks admit explicit combinatorial representations of their Laplace-transformed occupation probabilities and hitting-time distributions as sums over weighted spanning trees and forests of an augmented graph, \new{generalizing the previously developed graph-theoretic framework for Markov networks}. These representations yield asymptotic results, including the long-time scaling of hitting-time survival probabilities in the presence of infinite-mean sojourn times. They also extend recently derived combinatorial formulas for hitting-time moments to the semi-Markov setting, where non-Markovian corrections appear in higher moments. \new{For finite-mean sojourn times, the framework recovers the known stationary distribution through a generalization of the Markov Chain Tree Theorem, with transition rates replaced by pseudo-rates.} In particular, the mean hitting-time formula coincides with that of a Markov network with the corresponding pseudo-rates, which also yields a short \new{graph-theoretic proof} of the \new{previously established} constancy of Kemeny’s constant for finite-mean semi-Markov networks. Finally, we use the same framework to \new{extend established results on} weak ergodicity breaking for time-additive state observables that are nonconstant on the support of the limiting distribution in heavy-tailed semi-Markov networks, \new{allowing for generally heterogeneous, transition-dependent waiting-time distributions}.}


\maketitle

\section{Introduction}\label{sec1}
Markov networks describe memoryless transitions between discrete states, implying exponentially distributed dwell (or sojourn) times. This description is, e.g., natural for chemical reaction networks on a fundamental level~\cite{gillespie1992rigorous, gunawardena2012linear} but only retains validity for meta-stable coarse-grained states and a strong separation of timescales~\cite{prinz2011markov}. However, for many complex systems of interest, this assumption fails and sojourn time distributions clearly deviate from exponential laws~\cite{niemann2016renewal}. Examples from biophysics include ion-channel opening~\cite{goychuk2002ion}, cell-cycle progression~\cite{jolly2022cycleflow} and neural dynamics~\cite{chen2009discrete}. \new{The cell cycle provides a particularly instructive example where a detailed Markovian description of the intricate biochemical reaction network underlying cell division would require solving a complex chemical master equation with many \textit{a priori} unknown rate constants, while a semi-Markov description in terms of a few coarse-grained states, such as the phases of the cell cycle, can successfully describe experimental data~\cite{jolly2022cycleflow}.}

Moreover, semi-Markov networks also provide an important modeling framework, e.g., for epidemiology~\cite{worthington2025evaluating} or credit-risk~\cite{d2006homogeneous}.

A classical result in the theory of finite Markov networks is the Markov Chain Tree Theorem (MCTT)~\cite{anantharam1989proof,pitman2018tree}, which states that for an irreducible network, the stationary occupation probability of a state $i$ can be expressed as
\begin{align}
    \pi_i=\frac{\sum_{\mathcal{T}\in T_{[i]}}w(\mathcal{T})}{\sum_{\mathcal{T}\in T}w(\mathcal{T})},\label{eq:MCTT}
\end{align}
where $T$ denotes the set of spanning trees of the network interpreted as a directed graph $\graph=(V(\graph),E(\graph))$ and $T_{[i]}$ is the set of  spanning trees with root $i\in V(\graph)$. For continuous-time Markov networks, the weights $w(\mathcal{T})$ are defined as products over the transition rates $k_{ij}$ on the tree:
\begin{align}
    w(\mathcal{T})=\prod_{(i,j)\in E(\mathcal{T})}k_{ij}.
\end{align}
This identity has proven useful for deriving structural statements about the stationary state of finite Markov networks, for example, in the theory of random perturbations and in the mathematical description of biochemical networks~\cite{schnakenberg1976network,freidlin1998random,mirzaev2013laplacian,owen2023size,bebon2026mutual}. Related expressions also exist for stationary currents and mean hitting times and have been developed into a thermodynamic calculus for nonequilibrium steady states by Khodabandehlou and collaborators~\cite{khodabandehlou2022trees,khodabandehlou2023drazin,khodabandehlou2024vanishing,khodabandehlou2024poisson}. While such a graph-theoretic approach to Markov networks remained widely regarded as specific for the stationary case, recent work has extended the framework to transient probabilities and hitting-time densities~\cite{nam2025algebraic,voits2026emergence,voits2026graph}. A key observation to derive this extension is that the resolvent $(s\mathbb{1}-W)^{-1}$ of a Markov network\new{, where
\begin{align}
    W_{ij}:=\begin{cases}
        -\sum_{l\in V(\graph)}k_{il}&\text{if }i=j\\
        \qquad~~~ k_{ji}&\text{if }i\neq j,
    \end{cases}
\end{align}
and $s>0$} permits an explicit combinatorial representation by identifying $s\mathbb{1}-W$ as the principal minor of an augmented graph as a result of the All-Minors Matrix-Tree Theorem~\cite{chaiken1982combinatorial}. This allows for an explicit expression of the Laplace transform of the occupation probability $\tilde{p}_i(s)$ \new{for $s>0$} of a state $i\in V(\graph)$ in terms of spanning trees and forests of the augmented graph $G^\ast=\big(V(\graph)\cup\{\ast\},E(\graph)\cup\{(n,\ast)~|~n\in V(\graph)\}\big)$, similarly to Eq.~(\ref{eq:MCTT}):
\begin{align}
    \tilde{p}_i(s)&=\sum_{j\in V(\graph)}\frac{\sum_{\mathcal{F}^{s,j\to i}_{[i,\ast]}}w(\mathcal{F})}{\sum_{\mathcal{T}^{s}_{[\ast]}}w(\mathcal{T})}p_{j}(0),\label{eq:introduction_transient_occupation_probabilities}
\end{align}
where the term in the denominator sums over the spanning trees of the augmented network with root $\ast$, and the numerator spans over two-component spanning forests with roots $i$ and $\ast$, where the vertex $j$ is included in the component containing $i$. \new{The superscript $s$ is used to indicate that the edges leading to $\ast$ are given the weight $s$.}
Analogously, the hitting-time density $\tilde{f}_{i\to N}(s)$ from state $i\in V(\graph)$ to state $N\in V(\graph)$ can be expressed as
\begin{align}
    \tilde{f}(s)&=p_{0,N}+\sum_{j=1}^{N-1}\frac{\sum_{i=1}^{N-1}k_{iN}\sum_{\mathcal{F}^{s,j\to i}_{[i,N,\ast]}}w(\mathcal{F})}{\sum_{\mathcal{F}^{s}_{[N,\ast]}}w(\mathcal{F})}p_{j}(0).\label{eq:introduction_FPT_distribution}
\end{align}

In this work, the graph-theoretic framework is extended to semi-Markov networks, i.e., to the case of general waiting-time-dependent transitions. Replacing the constant jump rates in the Markovian case by memory kernels in Laplace space, $k_{ij}\mapsto\tilde{K}_{ij}(s)$, Eqs.~(\ref{eq:introduction_transient_occupation_probabilities}) and (\ref{eq:introduction_FPT_distribution}) are shown to retain validity for semi-Markov networks (Theorem \ref{theorem:transient_semi_MCTT} and Proposition \ref{prop_fpt_distribution_semi_Markov}, respectively). In particular, for finite-mean waiting times, the stationary occupation probabilities have the same tree representation as in Eq.~(\ref{eq:MCTT}), with the transition rates replaced by pseudo-rates $\kappa_{ij}=\frac{p_{ij}}{\mathbb{E}[\tau_i]}$, where $p_{ij}$ is the probability for a transition from $i$ to $j$ and $\mathbb{E}[\tau_i]$ is the mean sojourn time in state $i$ (Theorem \ref{theorem:SMCTT}). We also show that, for infinite-mean sojourn times with power-law tails, an analogous graph-theoretic identity holds, but only the heaviest-tailed states that can serve as roots of spanning trees of the underlying graph retain positive occupation probability in the  $t\to\infty$-limit. 

The long-time scaling of the hitting-time distribution is also dominated by the heaviest tails in the presence of infinite-mean sojourn times, and in the finite-mean case, the graph-theoretic formula for the mean hitting time again coincides with the known expression for Markov networks with the pseudo-rates, while higher moments involve additional terms which vanish in the Markovian case (Theorem \ref{theorem_fpt_semi_Markov}). Moreover, the expressions for the stationary occupation probabilities and the mean hitting times allow for a simple proof that Kemeny's constant remains independent of the starting state in semi-Markov networks with finite-mean sojourn times, including an explicit graph-theoretic representation~\cite{pitman2018tree}, a result that was previously obtained by algebraic means~\cite{hunter2018kemeny}. \new{For finite-mean sojourn times, this framework recovers established results for the stationary distribution, mean hitting times and Kemeny's constant of semi-Markov processes, while providing explicit graph-theoretic representations and alternative combinatorial derivations.}

Finally, the graph-theoretic framework for finite semi-Markov networks is applied to weak ergodicity breaking. The concept of ergodicity is central to the foundations of statistical mechanics and the idea can be traced back to the works of Boltzmann and others in the late nineteenth century~\cite{gallavotti1995ergodicity}. Ergodicity theory then became a mathematical discipline in its own right with several applications in statistical physics and beyond~\cite{einsiedler2011ergodic}.  
In comparison, the study of weak ergodicity breaking is relatively recent: in the 1990s, it was realized that semi-Markov networks with heavy-tailed waiting times can give rise to time-averaged occupation fractions which converge to the occupation probabilities on average but not almost surely~\cite{bouchaud1992weak}. This notion was investigated more thoroughly in the context of continuous-time random walks in the 2000s by Bel, Barkai, Metzler and others~\cite{bel2005weak,he2008random}. \new{In particular, Rebenshtok and Barkai derived a general limiting distribution for time-averaged observables in the presence of heavy-tailed waiting times, including an explicit relation between the fluctuations of time and ensemble averages~\cite{rebenshtok2007distribution,rebenshtok2008weakly}.} \newnew{The distributions of occupation time fractions in these systems is closely related to Lamperti's generalized arcsine law~\cite{lamperti1958occupation} and to more general limit theorems developed in ergodicity theory~\cite{thaler2006distributional}}. 

A time-additive state observable
\begin{align}
    A(t)=\integ{0}{t}{g(X_{t'})}{t'},
\end{align}
where $X_t$ is the stochastic process on the semi-Markov network and $g:V(\graph)\to \mathbb{R}$ is an arbitrary function assigning weights to the different states, exhibits weak ergodicity breaking if $\frac{\mathbb{E}[A(t)]}{t}$ converges to the expectation value of $g$ in the limiting state distribution for long times but $\frac{A(t)}{t}$ does not converge almost surely~\cite{bouchaud1992weak}. \new{Here, we recover the known variance relation for time-averaged state observables within the graph-theoretic framework and establish its validity for finite semi-Markov networks with generally heterogeneous, transition-dependent waiting-time distributions under suitable small-$s$ scaling conditions on the memory kernels (Theorem \ref{theorem:weak_ergodicity_breaking})}. 

The memoryless property provides a rich mathematical structure for studying Markov networks in terms of semigroups, spectral theory, and martingale methods~\cite{norris1998markov}, but this structure does not directly carry over to the semi-Markov case. In contrast, the graph-theoretic framework generalizes rather naturally to semi-Markov networks and allows us to prove properties of these systems on an elementary, combinatorial basis. These tools may help deepen our understanding of semi-Markov networks in the future.

\section{Semi-Markov networks and generalized master equation}
\label{sec:basics}

In this section, we introduce basic concepts for the theory of semi-Markov networks. The results are all standard~\cite{cinlar1975introduction} but helpful for establishing the conventions for the remainder of this work. We start with the definition of a semi-Markov network:

\begin{definition}
    Let $\mathcal{G}=(V(\graph),E(\graph))$ be a finite directed graph, with vertex set $V(\mathcal{G})=\{1,...,N\}$ for $N\in\mathbb{N}$ and edge set $E(\graph)\subseteq \{(i,j)\in V(\graph)\times V(\graph)~|~i\neq j\}$. A \textit{state} $n\in V(\mathcal{G})$ is called \textit{absorbing} if $\forall m\in V(\mathcal{G}):(n,m)\notin E(\graph)$ and \textit{non-absorbing} otherwise. We denote the set of non-absorbing states as $V^\circ(\graph)\subseteq V(\graph)$. A finite \textit{semi-Markov network} on $\graph$ is specified by a family of finite non-negative measures $\big(Q_{ij}\big)_{(i,j)\in E(\graph)}$ on $[0,\infty)$ with $Q_{ij}(\{0\})=0$ for all $(i,j)\in E(\graph)$ such that the quantities
    \begin{align}
      p_{ij}:=Q_{ij}([0,\infty))>0
    \end{align}
    satisfy for every $i\in V^\circ(\graph)$:
    \begin{align}
        \sum_{j:(i,j)\in E(\graph)}p_{ij}=1.
    \end{align}
    The measure $Q_{ij}(dt)$ is interpreted as the joint probability that, conditional on entering state $i$, the next state is $j$ and the sojourn time lies in $dt$. 
\end{definition}
Note that the restriction $Q_{ij}(\{0\})=0$ excludes atoms at $t=0$ in order to avoid instantaneous jumps.

We use the following notation throughout:
For $i\in V^\circ(\graph)$, define the probability measure
\begin{align}
    Q_i:=\sum_{j:(i,j)\in E(\graph)}Q_{ij},
\end{align}
which describes the distribution of a \textit{sojourn time} $\tau$ in state $i$ and satisfies
\begin{align}
    Q_i([0,\infty))=\sum_{j:(i,j)\in E(\graph)}p_{ij}=1.
\end{align}
The \textit{survival probability} in $i$ is denoted as
\begin{align}
    \phi_i(t):=Q_i((t,\infty)).
\end{align}
If $Q_i$ has a density $\psi_i=\sum_{j:(i,j)\in E(\graph)}\psi_{ij}$, this reduces to
\begin{align}
    \phi_i(t)=\integ{t}{\infty}{\psi_i(t')}{t'}.
\end{align}

We also introduce the \textit{memory kernels} $K_{ij}$ \new{for $(i,j)\in E(\graph)$}, defined via the following relation in Laplace space:
\begin{align}
    \tilde{K}_{ij}(s):=\frac{\tilde{Q}_{ij}(s)}{\tilde{\phi}_i(s)},
\end{align}
where
\begin{align}
    \tilde{Q}_{ij}(s):=\int_{[0,\infty)}e^{-st}Q_{ij}(dt)\label{eq:Def_Q}
\end{align}
denotes the Laplace-Stieltjes transform of $Q_{ij}$. For a locally integrable function $f$, we use the notation
\begin{align}
\tilde{f}(s)=\integ{0}{\infty}{f(t)e^{-st}}{t}
\end{align}
whenever the integral exists. More details on the Laplace-Stieltjes transform are given in the Appendix.

The Laplace transform of the survival probability is given by
\begin{align}
    \tilde{\phi}_i(s)&=\integ{0}{\infty}{e^{-st}Q_i((t,\infty))}{t}\\ &=\integ{0}{\infty}{e^{-st}\integ{[0,\infty)}{}{\mathbb{1}_{\{t'>t\}}Q_i(}{t')}}{t}\\ &=\integ{[0,\infty)}{}{\integ{0}{t'}{e^{-st}}{t}~Q_i(}{t')}\\ &=\integ{[0,\infty)}{}{\frac{1-e^{-st'}}{s}~Q_i(}{t')}\\ &=\frac{1-\tilde{Q}_i(s)}{s},
\end{align}
using that $Q_i([0,\infty))=1$ and the Fubini-Tonelli theorem, which allows us to exchange the integrals to obtain the third equality because the integrands are non-negative.

For absorbing states $i\in V(\graph)\setminus V^\circ(\graph)$, we set
\begin{align}
    \phi_i(t):=1,~~\tilde{Q}_{ij}(s):=0,~~\tilde{K}_{ij}(s):=0 \text{ for all $j\in V(\graph)$}.
\end{align}

Moreover, \new{we will use that the matrix
\begin{align}
    \tilde{W}_{ij}(s)&:=\begin{cases}
    -\sum_{l\in V(\graph)}\tilde{K}_{il}(s)&\text{if }i=j,\\
    \qquad\tilde{K}_{ji}(s)&\text{if }i\neq j,\label{eq:def_W_matrix}
    \end{cases}
\end{align}
is of the shape of a Laplacian matrix of a weighted directed graph, as explained in Sec.\ref{sec:graph_theoretic_structure}. Also note that setting $\tilde{Q}_{ij}(s):=0$ for $(i,j)\notin E(\graph)$ (and $\tilde{Q}_{ij}(s)$ as in Eq.~(\ref{eq:Def_Q}) otherwise) turns $\tilde{Q}(s)$ into the adjacency matrix of a weighted and directed graph}. For convenience, we define
\begin{align}
    \tilde{Q}_i(s):=\int_{[0,\infty)}e^{-st}Q_i(dt)
    =\sum_{j:(i,j)\in E(\graph)}\tilde{Q}_{ij}(s).
\end{align}
The two matrices are connected as follows:
\begin{lemma}
    Let $\diag{\tilde{\phi}(s)}:=\diag{\tilde{\phi}_1(s),..., \tilde{\phi}_N(s)}$ and $\tilde{W}(s)$ and $\tilde{Q}(s)$ as defined above. Then the following identity holds:
    \begin{align}
        (s\mathbb{1}-\tilde{W}(s))\diag{\tilde{\phi}(s)}=\mathbb{1}-\tilde{Q}^T(s),
    \end{align}
    where $\mathbb{1}$ denotes the identity matrix in $N$-dimensions.
    
    \noindent Moreover, $\diag{\tilde{\phi}(s)}$ is invertible, so also
    \begin{align}
        s\mathbb{1}-\tilde{W}(s)=\big(\mathbb{1}-\tilde{Q}^T(s)\big)\diag{\tilde{\phi}(s)}^{-1}.
    \end{align}
    \label{lemma:Q_W_identity}
\end{lemma}
\new{\begin{proof}
    This is a standard relation~\cite{cinlar1975introduction}. For completeness, a proof is provided in Appendix \ref{appendix_proofs}.
\end{proof}}

For a given semi-Markov network on $\mathcal{G}$ and an initial distribution $(p_{0,i})_{i\in V(\graph)}\in[0,1]$ with $\sum_{i\in V(\mathcal{G})}p_{0,i}=1$, consider the discrete-time Markov process $(X^d_n)_{n\in\mathbb{N}_0}$ with transition matrix
\begin{align}
     P_{ij}:=\begin{cases}
            p_{ij}&\text{if}~(i,j)\in E(\mathcal{G})\\1&\text{if}~i=j\in V(\graph)\setminus V^\circ(\graph)\\ 0&\text{else},
        \end{cases}
\end{align}
initially distributed according to
$(p_{0,i})_{i\in V(\graph)}$. Conditional on the discrete trajectory, the sojourn times $(\tau_n)_{n\in\mathbb N_0}$ are sampled independently such that, conditional on $X_n^d=i$ and $X_{n+1}^d=j$, the distribution of $\tau_n$ is $\frac{Q_{ij}}{p_{ij}}$, i.e.,
\begin{align}
    \mathbb{P}\left(\tau_n\in B\mid X_n^d=i,X_{n+1}^d=j\right)
    =\frac{Q_{ij}(B)}{p_{ij}}
\end{align}
for every Borel set $B\subseteq[0,\infty)$. If $X_n^d$ is absorbing, we may set $\tau_n=\infty$, or equivalently stop the sequence of jumps. We set $t_0:=0$ and
\begin{align}
    t_n:=\sum_{k=0}^{n-1}\tau_k
\end{align}
for $n\geq 1$. Then, the corresponding \textit{semi-Markov process} is defined as
\begin{align}
    (X_t,\tau):=(X^d_{\max\{n\in \mathbb{N}_0 ~|~t_n~\leq~t\}},t-t_{\max\{n\in \mathbb{N}_0 ~|~t_n~\leq~t\}}),
\end{align}
where $X_t$ is the state at time $t$ and $\tau$ is the time since the last transition, with the convention that after an absorbing state is reached the process remains there forever. The \textit{occupation probability} of state $i\in V(\graph)$ is denoted as $p_i(t):=\mathbb P(X_t=i)$.

The occupation probabilities of semi-Markov processes obey a generalized master equation in Laplace space. The proof is standard, but the form is instructive for the graph-theoretic formulation developed below. 
\begin{proposition}
    Let $X_t$ be a semi-Markov process on a network on $\graph=\big(V(\graph), E(\graph)\big)$ with $\tilde{W}(s)$ as defined in Eq.~(\ref{eq:def_W_matrix}) with \new{an arbitrary} initial distribution $p_0=(p_{0,i})_{i\in V(\graph)}$ \new{with $p_{0,i}\geq 0$ for all $i\in V(\graph)$ and $\sum_{i\in V(\graph)}p_{0,i}=1$}. Then the Laplace transform of the occupation probabilities $p(t)=\big(p_i(t)\big)_{i\in V(\graph)}$ obeys the following relation:  
    \begin{align}
        (s\mathbb{1}-\tilde{W}(s))\tilde{p}(s)=p_0.\label{eq:generalized_ME_Laplace}
    \end{align}
    \label{prop:generalized_ME_Laplace}
\end{proposition}
\new{\begin{proof}
    This is a standard relation~\cite{cinlar1975introduction}. For completeness, a proof is provided in Appendix \ref{appendix_proofs}.
\end{proof}}

Under additional regularity assumptions, or more generally in the
sense of distributions, Eq.~(\ref{eq:generalized_ME_Laplace}) corresponds to an integro-differential equation:
\begin{align}
    \tder p_i(t)&=\integ{0}{t}{\sum_{j\in V(\graph)}W_{ij}(t-t')p_j(t')}{t'}\\ &=\integ{0}{t}{\sum_{j\in V(\graph)\setminus\{i\}}K_{ji}(t-t')p_j(t')}{t'}-\integ{0}{t}{\sum_{l\in V(\graph)}K_{il}(t-t')p_i(t')}{t'},
\end{align}
known as the generalized master equation~\cite{kenkre1973generalized}. In particular, for exponentially distributed transition times, $Q_{ij}(dt)=k_{ij}\exp{(-\sum_l k_{il}t)}~\text{d}t$ for $k_{ij}>0$, the Laplace transforms read: 
\begin{align}
    \tilde{Q}_{ij}(s)=\frac{k_{ij}}{s+\sum_l k_{il}},
\end{align}
which means that the Laplace transform of the memory kernels is constant:
\begin{align}
    \tilde{K}_{ij}(s)=\frac{s\tilde{Q}_{ij}(s)}{1-\tilde{Q}_i(s)}=\frac{s\frac{k_{ij}}{s+\sum_l k_{il}}}{1-\sum_l\frac{k_{il}}{s+\sum_l k_{il}}}=k_{ij},
\end{align}
corresponding to $k_{ij}\delta(t)$ in time domain. Then, $\tilde{W}(s)$ turns into the generator of a Markov process and the usual master equation is recovered:
\begin{align}
    \tder p_i(t)&=\sum_{j\in V(\graph)\setminus\{i\}}k_{ji}p_j(t)-\sum_{l\in V(\graph)}k_{il}p_i(t).
\end{align}

\section{Graph-theoretic structure}\label{sec:graph_theoretic_structure}
The combinatorial structure of the occupation probabilities and similar quantities of interest follows from identifying $s\mathbb{1}-\tilde{W}(s)$ and $\tilde{Q}(s)$ with the Laplacian matrix and the adjacency matrix of a weighted and directed graph, respectively. We employ the following definitions and conventions:
\begin{definition}
    Let $\graph=(V(\graph),E(\graph))$. For $i,j\in V(\graph)$, a (non-self-intersecting) \textit{path from i to j}, $\mathcal{P}^{i\to j}\subseteq\graph$, is a subgraph of the shape:
    \begin{align}
        \mathcal{P}^{i\to j}=\big(\{i, l_1,...,l_{n-1},j\},\{(i, l_1), (l_1,l_2),...,(l_{n-1},j)\}\big).
    \end{align}
    This includes the trivial path $\mathcal{P}^{i\to i}=(\{i\},\emptyset )$. If such a path $\mathcal{P}^{i\to j}$ exists, we say that $i$ is \textit{connected to} $j$ (in the edge-directed sense).

    Given a set of \textit{roots} $I=\{i_1,...,i_n\}\subseteq V(\graph)$, $I\neq \emptyset$, a \textit{spanning forest} 
    \begin{align}
        \mathcal{F}_{[i_1,...,i_n]}=\big(V(\graph),E(\mathcal{F}_{[i_1,...,i_n]})\big)\subseteq\graph
    \end{align}
    is a subgraph, where for every $j\in V(\graph)$, there is exactly one path $\mathcal{P}^{j\to i}$ from $j$ to exactly one $i\in I$. The set of all spanning forests of $\graph$ with root set $I=\{i_1,...,i_n\}$ is denoted as $F_{[i_1,...,i_n]}$.
    Moreover, a spanning forest with exactly one root $i\in V(\graph)$ is called a \textit{spanning tree}. Instead of $\mathcal{F}_{[i]}$ and $F_{[i]}$, we then write $\mathcal{T}_{[i]}$ and $T_{[i]}$ to emphasize this particular structure.
\end{definition}
By convention, we set $F_{[i_1,...i_n]}=\emptyset$, if $i_l=i_m$ for some $l\neq m$.
\new{Here and in the following, directed trees and forests are always understood inward-oriented, i.e., with paths leading to the root, in contrast to the outward convention, which is also found in the literature.}

We often need to restrict a set of spanning forests to a subset containing a path between two vertices $l,m\in V(\graph)$, which we denote as follows:
\begin{align}
    F_{[i_1,...,i_n]}^{l\to m}:=F_{[i_1,...,i_n]}\cap \{\mathcal{S}\subseteq \graph~|~\mathcal{S}\text{ contains a $l$ to $m$ path} \}.
\end{align}

An illustration of these concepts is given in Fig.~\ref{fig:graph_theory_basics}: an example of a (directed) graph is shown in Fig.~\ref{fig:graph_theory_basics}(a) and a path as a subgraph of it is shown in Fig.~\ref{fig:graph_theory_basics}(b). Fig.~\ref{fig:graph_theory_basics}(c) shows an example of a spanning forest with two roots containing this path and Fig.~\ref{fig:graph_theory_basics}(d) shows an example of a spanning tree, i.e., a spanning forest with exactly one root.
\begin{figure}[t]
    \centering
    \includegraphics[width=1\linewidth]{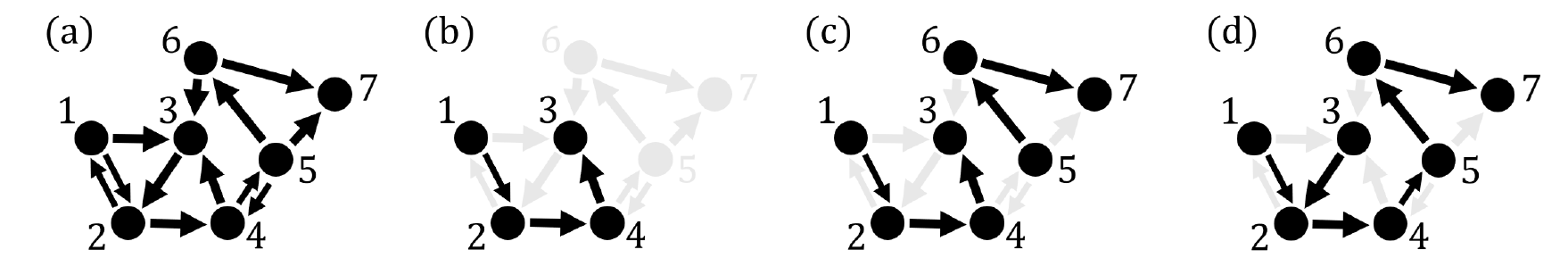}
    \caption{Illustration of graph-theoretic concepts. (a) A directed graph $\graph=(V(\graph),E(\graph))$ with $7$ vertices. (b) A path $\mathcal{P}^{1\to 3}\subset\graph$ from $1$ to $3$. (c) A spanning forest $\mathcal{F}^{1\to3}_{[3,7]}\subset\graph$ with roots $3$ and $7$ containing a path from $1$ to $3$. (d) A spanning tree $\mathcal{T}_{[7]}\subset\graph$ with root $7$.}
    \label{fig:graph_theory_basics}
\end{figure}
\begin{definition}
    Let $\graph =\big(V(\graph), E(\graph)\big)$ be a directed graph, and let
    \begin{align}
        w:E(\graph)&\to (0,\infty)\\ (i,j)&\mapsto w(i,j)
    \end{align}
    be a weight function that assigns a weight $w(i,j)>0$ to every $(i,j)\in E(\graph)$. The \textit{adjacency matrix} $A\in\mathbb{R}^{N\times N}$ of the weighted graph $(\graph,w)$ is defined as:
    \begin{align}
        A_{ij}:=\begin{cases}
            w(i,j)&\text{if }(i,j)\in E(\graph)\\~~~~0&\text{else},
        \end{cases}
    \end{align}
    and the \textit{(out-degree) Laplacian matrix} $L$ is defined as:
    \begin{align}
        L_{ij}:=\begin{cases}
            ~~ -A_{ij}&\text{if }i\neq j\\\sum_{l\in V(\graph)}A_{il}&\text{if }i=j,
        \end{cases}.
    \end{align}
\end{definition}

For a given weight function $w$, we define the weight of a subgraph $\mathcal{S}=\big(V(\mathcal{S}),E(\mathcal{S})\big)\subseteq \graph$ to be the product of the weights of its edges, i.e.,
\begin{align}
    w(\mathcal{S}):=\prod_{(i,j)\in E(\mathcal{S})}w(i,j),
\end{align}
with the convention that $w((V(\graph),\emptyset))=1$. The All-Minors Matrix-Tree Theorem allows us to express the determinants of minors of the Laplacian matrix of a weighted and directed graph in terms of sums over the weight of spanning forests. It was first proved in full generality by Chaiken~\cite{chaiken1982combinatorial}. We do not state the full theorem here but only the following special cases relevant for our purposes:

\begin{theorem}
    Let $(\graph,w)$ be a directed and weighted graph with vertex set $V(\graph)=\{1,...,N\}$. Let $L$ be its Laplacian matrix, $i,j\in V(\graph)\setminus\{N\}$ and let $L_{\{N\},\{N\}}$ and $L_{\{i,N\},\{j,N\}}$ be the $N$-th principal  minor and the minor obtained from deleting the rows $i$ and $N$ and the columns indexed by $j$ and $N$, respectively. Then the following identity holds for their determinants:
    \begin{align}
        \det{L_{\{N\},\{N\}}}&=\sum_{\mathcal{T}\in T_{[N]}}w(\mathcal{T}),\label{eq_n-th-principal-minor}\\
        \det{L_{\{i,N\},\{j,N\}}}&=(-1)^{i+j}\sum_{\mathcal{F}\in F^{j\to i}_{[i,N]}}w(\mathcal{F}).\label{eq:_n_i_n_j-minor}
    \end{align}
    \label{theorem:AMMTT}
\end{theorem}
\begin{proof}
    \new{This theorem is a result obtained by Chaiken~\cite{chaiken1982combinatorial}}. The key idea is to expand the determinant and to interpret the resulting terms as weights of directed subgraphs. One then shows that precisely the terms containing a cycle cancel by sign-reversing pairing, while the acyclic terms correspond exactly to the spanning forests.
\end{proof}

For notational convenience, we also introduce the convention that instead of writing $\mathcal{T}\in T_{[n]}$ and $\mathcal{F}\in F^{j\to i}_{[i,n]}$ in the sums, we simply write the indices as $\mathcal{T}_{[n]}$ and $\mathcal{F}^{j\to i}_{[i,n]}$, understanding that the sum runs over the corresponding sets. For example, we would write $\sum_{\mathcal{F}^{j\to i}_{[i,n]}}w(\mathcal{F})$ instead of $\sum_{\mathcal{F}\in F^{j\to i}_{[i,n]}}w(\mathcal{F})$.

The matrix $s\mathbb{1}-\tilde{W}(s)$ itself is not a Laplacian, but can be identified as a principal minor of a Laplacian by adding an auxiliary vertex $\ast$ to the original graph, connecting all other vertices to it, and assigning the weight $s$ to each additional edge. We shall refer to this construction as the \textit{augmented graph} $\graph^\ast$ of $\graph$. Then, Eqs.~(\ref{eq_n-th-principal-minor}) and (\ref{eq:_n_i_n_j-minor}) allow us to give an explicit combinatorial expression for the inverse of $s\mathbb{1}-\tilde{W}(s)$:

\begin{proposition}
    Let $\graph=\big(V(\graph),E(\graph)\big)$ with $V(\graph)=\{1,...,N\}$ be a semi-Markov network with $\tilde{W}(s)$ as defined in Eq.~(\ref{eq:def_W_matrix}). Then, for every $s>0$, $\big(s\mathbb{1}-\tilde{W}(s)\big)^T$ is a principal minor of a Laplacian matrix of an \textit{augmented graph}:
    \begin{align}
        \mathcal{G}^\ast=\big(V(\graph^\ast),E(\graph^\ast)\big) :=\big(V(\graph)\cup\{\ast\},E(\mathcal{\graph})\cup\{(i,\ast)~|~i\in V(\graph)\}\big),
    \end{align}
    with the weight function defined by:
    \begin{align}
        w:E(\graph^\ast)&\to(0,\infty)\\ (i,j)&\mapsto\begin{cases}
            ~~~~s&\text{if }j=\ast\\ \tilde{W}_{ji}(s)&\text{else}.
        \end{cases}
    \end{align}
    Moreover, $s\mathbb{1}-\tilde{W}(s)$ is invertible and its inverse reads:
    \begin{align}
        \big(s\mathbb{1}-\tilde{W}(s)\big)^{-1}_{ij}=\frac{\sum_{\mathcal{F}^{s, j\to i}_{[i,\ast]}}w(\mathcal{F})}{\sum_{\mathcal{T}^{s}_{[\ast]}}w(\mathcal{T})},
    \end{align}
    where $\mathcal{T}^s$ and $\mathcal{F}^s$ indicate the fact that these are the spanning trees and forests of the augmented network $\graph^\ast$.\label{Prop:inverse_of_s-W}
\end{proposition}
The construction of the augmented graph $G^\ast$ is illustrated in Fig.~\ref{fig:augmented_graph}.
\begin{figure}[t]
    \centering
    \includegraphics[width=0.8\linewidth]{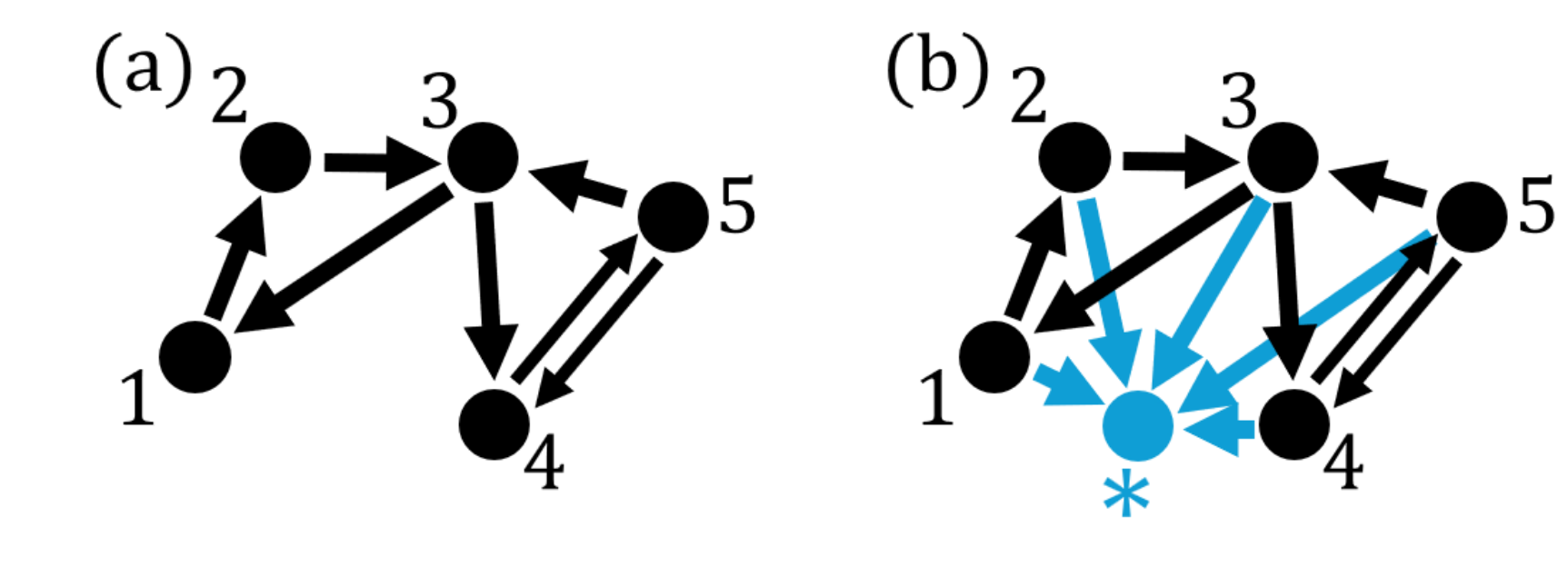}
    \caption{Illustration of the construction of the augmented graph. (a) A base graph $\graph =(V(\graph),E(\graph))$ with $V(\graph)=\{1,2,3,4,5\}$ and $E(\graph)=\{(1,2),(2,3),(3,1),(3,4),(4,5),(5,3), (5,4)\}$. (b) The augmented graph $\graph^\ast=(V(\graph^\ast),E(\graph^\ast))$ obtained from $\graph$ by adding an extra vertex $\ast$ and connecting all the other vertices in $V(\graph)$ to it, i.e., $V(\graph^\ast)=V(\graph)\cup\{\ast\}$ and $E(\graph^\ast)=E(\graph)\cup\{(i,\ast)~|~i\in V(\graph)\}$.}
    \label{fig:augmented_graph}
\end{figure}

\begin{proof}
    \new{The key ideas for the proof are similar to the Markovian case, for which the relation was already shown in \cite{voits2026graph}}.
    
    The augmented network is obtained by adding a vertex $\ast$ to $V(\graph)$ and connecting all other vertices to it. The weights for the edges $(i,j)\in E(\graph)$ are $\tilde{W}_{ji}(s)=\tilde{K}_{ij}(s)$ and $w(i,\ast)=s$ for $i\in V(\graph)$. Hence, its Laplacian matrix reads for $n,m\in V(\graph^\ast)$:
    \begin{align}
        L_{nm}=\begin{cases}
            ~~~-\tilde{W}_{mn}(s)& \text{for }n,m\in V(\graph)~n\neq m,\\
            \qquad-s& \text{for }n\in V(\graph), m=\ast,\\
            \qquad ~~0& \text{for }n=\ast, m\in V(\graph^\ast),\\ 
            s+\sum_{l\in V(\graph)\setminus\{n\}}\tilde{W}_{ln}(s)& \text{for }m=n\neq\ast.
        \end{cases}
    \end{align}
    Therefore, the $\ast$-principal minor of $L$ reads:
    \begin{align}
        \big(L_{\{\ast\},\{\ast\}}\big)_{nm}&=\begin{cases}
            ~~~-\tilde{W}_{mn}(s)& \text{for }n,m\in V(\graph)~n\neq m,\\
            s+\sum_{l\in V(\graph)\setminus\{n\}}\tilde{W}_{ln}(s)& \text{for }m=n\neq\ast.
        \end{cases}\\ &=(s\mathbb{1}-\tilde{W}(s))_{mn}\\ &=\big((s\mathbb{1}-\tilde{W}(s))^T\big)_{nm}.
    \end{align}
    So we see that for $s>0$:
    \begin{align}
        \det(s\mathbb{1}-\tilde{W}(s))=\det L_{\{\ast\},\{\ast\}}=\sum_{\mathcal{T}^{s}_{[\ast]}}w(\mathcal{T})\geq s^N>0,
    \end{align}
    since the augmented graph always contains the spanning tree where each original vertex points directly to $\ast$ and all the other weights are non-negative. We can apply Theorem \ref{theorem:AMMTT} to obtain the inverse of $s\mathbb{1}-\tilde{W}(s)$:
    \begin{align}
        \big(s\mathbb{1}-\tilde{W}(s)\big)^{-1}_{ij}&=\frac{(-1)^{i+j}\det\big(\big(s\mathbb{1}-\tilde{W}(s)\big)_{\{j\}\{i\}}\big)}{\det\big(s\mathbb{1}-\tilde{W}(s)\big)}\\ &=\frac{(-1)^{i+j}\det L_{\{i,\ast\}\{j,\ast\}}}{\det L_{\{\ast\},\{\ast\}}}\\ &=\frac{\sum_{\mathcal{F}^{s,j\to i}_{[i,\ast]}}w(\mathcal{F})}{\sum_{\mathcal{T}^{s}_{[\ast]}}w(\mathcal{T})}.
    \end{align}
\end{proof}

Note that removing the auxiliary vertex again turns a tree with root $\ast$ in $\graph^\ast$ into a forest in $\graph$, where the roots are the vertices that previously had an outgoing edge to $\ast$. Therefore, the sums over the spanning trees and the spanning forests with two trees in $\graph^\ast$ may also be written in terms of spanning forests of $\graph$:
\begin{align}
    \big(s\mathbb{1}-\tilde{W}(s)\big)^{-1}_{ij}&=\frac{\sum_{\mathcal{F}^{s,j\to i}_{[i,\ast]}}w(\mathcal{F})}{\sum_{\mathcal{T}^{s}_{[\ast]}}w(\mathcal{T})}\\ &=\frac{\sum_{n=0}^{\new{N-2}\newnew{+\delta_{ij}}}s^n\sum_{\{l_1,...l_n\}\subseteq V(\graph)\setminus\{i,j\}}\sum_{\mathcal{F}^{j\to i}_{[l_1,...,l_n,i]}}w(\mathcal{F})}{\sum_{n=1}^Ns^{n}\sum_{\{l_1,...l_n\}\subseteq V(\graph)}\sum_{\mathcal{F}_{[l_1,...l_n]}}w(\mathcal{F})}.
\end{align}

As a corollary of the previous Proposition, we get:
\begin{theorem}\label{theorem:transient_semi_MCTT}
    Let $X_{t\geq 0}$ be a semi-Markov process on a network $\graph$ with initial distribution $\big(p_0\big)_{j\in V(\graph)}$. The Laplace transform of the occupation probability $\tilde{p}_i(s)$ for state $i\in V(\graph)$ is given by:
    \begin{align}
        \tilde{p}_i(s)=\sum_{j\in V(\graph)}\frac{\sum_{\mathcal{F}^{s,j\to i}_{[i,\ast]}}w(\mathcal{F})}{\sum_{\mathcal{T}^{s}_{[\ast]}}w(\mathcal{T})}p_{0,j}\label{eq:transient_semi_MCTT}.
    \end{align}
\end{theorem}
\begin{proof}
    This is an immediate consequence of applying Proposition \ref{Prop:inverse_of_s-W} to Eq.~(\ref{eq:generalized_ME_Laplace}). 
\end{proof}

Eq.~(\ref{eq:transient_semi_MCTT}) can be seen as the transient, semi-Markov analogue of the classical Markov Chain Tree Theorem: stationary trees are replaced by $s$-dependent forest weights in Laplace space. Next, we will see that the stationary result follows from the $s\to 0$ limit, generalizing the Markov Chain Tree Theorem to semi-Markov networks.

\section{Stationary solution}
The final value theorem for Laplace transforms states that if the limit $\lim_{t\to \infty}f(t)$ and $\lim_{s\to 0}s \tilde{f}(s)$ exist, they are related as follows~\cite{doetsch2012introduction}:
\begin{align}
   \lim_{t\to \infty}f(t)=\lim_{s\to 0} s \tilde{f}(s),
\end{align}
see also the Appendix for further details on the asymptotic relation between the function in time domain and in Laplace space.

We want to apply this to the occupation probabilities $\tilde{p}_i(s)$ to find the limiting state distribution in time domain. The limit in Laplace space can be calculated using Eq.~(\ref{eq:transient_semi_MCTT}), but identifying the small-$s$ limit with a time-domain limit requires the latter to exist.
Moreover, the definition of semi-Markov networks requires that the transition-time distributions satisfy $\Tilde{Q}_{ij}(0)=p_{ij}$ but does not determine the small-$s$ scaling. To study the asymptotic properties, we restrict the discussion to transition-time distributions of the following shape:
\begin{assumption}[Small-$s$ regime of sojourn time distributions]\label{assumption:small_s_sojourn_times}
    For every non-absorbing state $i\in V^\circ(\graph)$, assume that the Laplace-Stieltjes transforms $\tilde{Q}_{i}(s)$ permit a small-$s$ expansion as follows:
    \begin{align}
         \Tilde{Q}_{i}(s)=1-a_{i}s^{\alpha_{i}}+o(s^{\alpha_{i}}),
    \end{align}
    where $a_{i}>0$ and $\alpha_{i}\in (0,1]$. For an absorbing state $i$, we may set $\alpha_i=0$.
\end{assumption}
The case $\alpha_{i}=1$ corresponds to sojourn times with finite mean. More precisely, for a sojourn time $\tau_{i}\sim Q_i$ in state $i\in V(\graph)$, we have
\begin{align}
    \mathbb{E}[\tau_{i}]<\infty\Leftrightarrow \Tilde{Q}_i(s)=1-\mathbb{E}[\tau_{i}]s+o(s).
\end{align}
In the case $0<\alpha_{i}<1$, the first moment is infinite. This regime describes power-law transition-time distributions. 

For the memory kernels, the small-$s$ regime then reads:
\begin{align}
    \tilde{K}_{ij}(s)=\frac{s\Tilde{Q}_{ij}(s)}{1-\Tilde{Q}_{i}(s)}=\frac{p_{ij}}{a_i}s^{1-\alpha_i}+o(s^{1-\alpha_i}),
\end{align}
since $\Tilde{Q}_{ij}(0)=p_{ij}$.

Moreover, it is convenient to introduce the \textit{pseudo-rates} $\kappa_{ij}:=\frac{p_{ij}}{a_i}$, which in the case of a Markov network reduce to the rates $k_{ij}$ of the exponentially distributed jumping times, and the weight of a subgraph $\mathcal{S}\subseteq \graph$ with respect to the pseudo-rates:
\begin{align}
    w_\kappa(\mathcal{S}):=\prod_{(i,j)\in E(\mathcal{S})}\kappa_{ij}.
\end{align}
We also write 
\begin{align}
    \alpha^\ast:=\min_{j\in\{1,...,N\}} \{\alpha_j~|~\exists~ \mathcal{T}\subseteq \graph:\mathcal{T}\text{ is a spanning tree of $\graph$ with root }j\}
\end{align}
and define the corresponding bottleneck vertices 
\begin{align}
    V_{\alpha^\ast}:=\{i\in V(\graph)~|~\alpha_i=\alpha^\ast~\land~\exists~ \mathcal{T}\subseteq \graph:\mathcal{T}\text{ is a spanning tree of $\graph$ with root }i\}.
\end{align}

This now allows us to formulate a generalization of the MCTT to semi-Markov networks of this shape. \new{For general semi-Markov networks, the existence of a limiting distribution in time domain is not guaranteed. We first show the corresponding limit in Laplace space and then discuss sufficient conditions for pointwise convergence in time.}

\begin{theorem}\label{theorem:SMCTT}
    Let $X_{t\geq 0}$ be a semi-Markov process on a network on $\graph$ with at least one spanning tree  and sojourn times that satisfy Assumption \ref{assumption:small_s_sojourn_times}. \new{Then, the limit $\pi :=\lim_{s\to 0}s\tilde{p}(s)$ exists, is unique in the sense that it is independent of the initial distribution,} and for $i\in V(\graph)$, it is given by:
    \begin{align}
    \pi_i=\begin{cases}
        \frac{\sum_{\mathcal{T}_{[i]}}w_\kappa(\mathcal{T})}{\sum_{l\in V_{\alpha^\ast}}\sum _{\mathcal{T}_{[l]}}w_\kappa(\mathcal{T})}&\text{if}~\alpha_i=\alpha^\ast,\\ \qquad~~0 &\text{else}.
    \end{cases}\label{eq:stationary_solution}
\end{align}
\end{theorem}
\begin{proof}
    Observe that $\graph$ can contain at most one absorbing vertex $n\in V(\graph)$ because it contains a spanning tree. If this is the case, all spanning trees have to be rooted at $n$, so Eq.~(\ref{eq:transient_semi_MCTT}) implies that:
    \begin{align}
        \pi_i=\lim_{s\to 0}s\tilde{p}_i(s)=\sum_{j\in V(\graph)}\frac{\sum_{\mathcal{T}_{[i]}}w(\mathcal{T})}{\sum_{\mathcal{T}_{[n]}}w(\mathcal{T})}p_{0,j}=\delta_{in}.
    \end{align}
    Since $\alpha_n=0$ and $\alpha_i>0$ for $i\in V(\graph)\setminus\{n\}$, Eq.~(\ref{eq:stationary_solution}) holds in this case.

    \new{Also note that if $\alpha_i<\alpha^*$ for some $i\in V(G)$, then, by definition of $\alpha^*$, $G$ cannot contain a spanning tree with root $i$. Since $G$ contains a spanning tree by assumption, there must be a path from $i$ to its root $j\in V(G)\setminus\{i\}$, but there cannot be a path from $j$ to $i$, as this would allow us to construct a spanning tree with root $i$. \newnew{Moreover, since every vertex in $V(G)$ is connected to $j$ and the sojourn times in non-absorbing states are almost surely finite, the hitting time $T_j$ of state $j$ is almost surely finite. As the process cannot reach $i$ after visiting $j$, we have
    \begin{align}
        p_i(t)\leq \mathbb{P}(T_j>t)\to 0
    \end{align}
    for $t\to\infty$ and hence $\pi_i=\lim_{s\to 0}s\tilde{p}_i(s)=0$.}  So we only need to show Eq.~(\ref{eq:stationary_solution}) for vertices $i\in V(\graph)$ with $\alpha_i\geq \alpha^\ast $.}
    
    Now consider the case that there are no absorbing vertices and observe that Assumption \ref{assumption:small_s_sojourn_times} implies the following small-$s$ scaling for a spanning forest $\mathcal{F}$ of $\graph$ with roots $\{i_1,...,i_n\}\subseteq V(\graph)$: 
    \begin{align}
    w(\mathcal{F})=w_\kappa(\mathcal{F})s^{|E(\mathcal{F})|-\sum_{k\in V(\graph)}\alpha_k+\sum_{k=1}^n\alpha_{i_k}}+o(s^{|E(\mathcal{F})|-\sum_{k\in V(\graph)}\alpha_k+\sum_{k=1}^n\alpha_{i_k}}).
    \end{align}
    Therefore, 
    \begin{align}
        \sum_{\mathcal{F}^{s,j\to i}_{[i,\ast]}}w(\mathcal{F})&=\sum_{\mathcal{T}_{[i]}}w_\kappa(\mathcal{T})s^{N-1-\sum_{k\in V(\graph)}\alpha_k+\alpha_{i}}+o(s^{N-1-\sum_{k\in V(\graph)}\alpha_k+\alpha_{i}}),\\\sum_{\mathcal{T}^{s}_{[\ast]}}w(\mathcal{T})&=\sum_{l:\alpha_l=\alpha^\ast}\sum _{\mathcal{T}_{[l]}}w_\kappa(\mathcal{T})s^{N-\sum_{k\in V(\graph)}\alpha_k+\alpha^\ast}+o(s^{N-\sum_{k\in V(\graph)}\alpha_k+\alpha^\ast}).
    \end{align}
    Note that the leading-order term of the sum over the spanning forests is independent of $j$.
    
    Eq.~(\ref{eq:transient_semi_MCTT}) implies:
    \begin{align}
        \pi_i&=\lim_{s\to 0}s\tilde{p}_i(s)\\&=\sum_{j\in V(\graph)}p_{0,j}\lim_{s\to0}\frac{\sum_{\mathcal{T}_{[i]}}w_\kappa(\mathcal{T})s^{\alpha_{i}}}{\sum_{l:\alpha_l=\alpha^\ast}\sum _{\mathcal{T}_{[l]}}w_\kappa(\mathcal{T})s^{\alpha^\ast}}\\ &=\begin{cases}
        \frac{\sum_{\mathcal{T}_{[i]}}w_\kappa(\mathcal{T})}{\sum_{l:\alpha_l=\alpha^\ast}\sum _{\mathcal{T}_{[l]}}w_\kappa(\mathcal{T})}&\text{if}~\alpha_i=\alpha^\ast,\\ \qquad~~0 &\text{else},
    \end{cases}
    \end{align}
    since $\sum_{j\in V(\graph)}p_{0,j}=1$. In particular, the limit is independent of the choice of the initial distribution $p_0$, i.e., \new{$\pi$} is uniquely determined by the network.
\end{proof}
In the case of Markov networks, or more generally, in the case that all sojourn times have finite mean, $\forall i\in V(\graph):\alpha_i=1$. Then, $\pi$ is the stationary distribution and we recover its shape from the Markov Chain Tree Theorem:
\begin{align}
    \pi_i=\frac{\sum_{\mathcal{T}_{[i]}}w_\kappa(\mathcal{T})}{\sum_{l\in V(\graph)}\sum _{\mathcal{T}_{[l]}}w_\kappa(\mathcal{T})}\label{eq:stationary_distribution_finite_mean_sojourn},
\end{align}
where $w_\kappa(\mathcal{T})=w(\mathcal{T})$ in the Markovian case.
Also note that for every spanning tree $\mathcal{T}_{[i]}\subseteq \graph$ with root $i\in V(\graph)$, we have
\begin{align}
    w_\kappa(\mathcal{T}_{[i]})=\frac{\mathbb{E}[\tau_i]}{\prod_{j\in V(\graph)}\mathbb{E}[\tau_j]}\prod_{(n,m)\in E(\mathcal{T}_{[i]})}p_{nm}.
\end{align}
Hence, using the classical Markov Chain Tree Theorem on the Markov chain with transition matrix $P$ to express the stationary distributions $(\nu_i)_{i\in V(\graph)}$ in terms of spanning trees,
\begin{align}
    \nu_i=\frac{\sum_{\mathcal{T}_{[i]}}\prod_{(n,m)\in E(\mathcal{T})}p_{nm}}{\sum_{l\in V(\graph)}\sum_{\mathcal{T}_{[l]}}\prod_{(n,m)\in E(\mathcal{T})}p_{nm}},
\end{align}
we recover the well-known expression for the stationary occupation probability for a semi-Markov network with finite-mean sojourn times:
\begin{align}
    \pi_i&=\frac{\sum_{\mathcal{T}_{[i]}}w_\kappa(\mathcal{T})}{\sum_{l\in V(\graph)}\sum _{\mathcal{T}_{[l]}}w_\kappa(\mathcal{T})}\\&=\frac{\mathbb{E}[\tau_i]\sum_{\mathcal{T}_{[i]}}\prod_{(n,m)\in E(\mathcal{T})}p_{nm}}{\sum_{l\in V(\graph)}\mathbb{E}[\tau_l]\sum _{\mathcal{T}_{[l]}}\prod_{(n,m)\in E(\mathcal{T})}p_{nm}}\\&=\frac{\mathbb{E}[\tau_i]\nu_i}{\sum_{l\in V(\graph)}\mathbb{E}[\tau_l]\nu_l}.
\end{align}

As stated previously, the existence of $\lim_{t\to\infty}p(t)$ is an additional assumption that has to be checked in the general case. But for a semi-Markov network with finite-mean sojourn times, it turns out that convergence to a stationary distribution is guaranteed if the transition-time distributions are nonlattice~\cite{athreya1978limit1,athreya1978limit}. Specifically, a probability measure $Q$ on $[0,\infty)$ is called lattice if there exist $a>0$ and $b\geq0$ such that
\begin{align}
    Q\big(\{b+na~|~n\in\mathbb{N}_0\}\big)=1.
\end{align}  
Intuitively, this means that transitions can happen only at discrete, equidistant points in time. Allowing for such transitions, it is easy to construct networks where the occupation probabilities do not settle into a stationary state but exhibit oscillatory behavior, such as in the following minimal example:

\begin{example}
    Consider the graph with two vertices $\graph=\big(\{1,2\},\{(1,2),(2,1)\}\big)$ and the semi-Markov network with $\delta$-like transition-time law:
    \begin{align}
        Q_{12}=\delta_\tau,~~Q_{21}=\delta_\tau.
    \end{align}
    Then the transitions are deterministic and it is clear that the occupation probabilities read:
    \begin{align}
        p_1(t)=1-p_2(t)=\begin{cases}
            p_{0,1}~~\text{if }2k\tau \leq t<(2k+1)\tau \text{ for }k\in \mathbb{N}_0,\\ p_{0,2}~~\text{if }(2k+1)\tau \leq t<2(k+1)\tau \text{ for }k\in \mathbb{N}_0. 
        \end{cases}
    \end{align}
    In particular, the limit $t\to \infty$ does not exist but the process alternates indefinitely between the two states. This is reflected by purely imaginary zeros of $\det(s\mathbb{1}-\tilde{W}(s))$. Note that $\tilde{K}_{12}(s)=\tilde{K}_{21}(s)=\frac{s}{e^{\tau s}-1}$, so:
    \begin{align}
         &0=\det\big(s\mathbb{1}-\tilde{W}(s)\big)=\sum_{\mathcal{T}^{s}_{[\ast]}}w(\mathcal{T}) =s^2\big(1+\frac{2}{e^{\tau s}-1}\big) ~\Leftrightarrow ~s=0~~\lor~~e^{\tau s}=-1.
    \end{align}
    Thus, $s=i \frac{2k+1}{\tau}\pi$ for $k\in \mathbb{Z}$ are also zeros of the determinant.
\end{example}

But ruling out lattice sojourn time distributions is sufficient to guarantee convergence to a stationary solution in semi-Markov networks with finite-mean sojourn times. In fact, the even weaker requirement that the return-time distribution to one state in the unique recurrent class be nonlattice would be sufficient~\cite{athreya1978limit,athreya1978limit1}. \new{For infinite-mean sojourn times, Assumption \ref{assumption:small_s_sojourn_times} determines the small-$s$ scaling of the occupation probabilities in Laplace space but does not by itself establish their pointwise convergence in time}

\begin{corollary}\label{cor:no_lattice_corollary}
    Let $X_{t\geq 0}$ be a semi-Markov process on a network on $\graph$ with finite-mean sojourn times, and let none of its sojourn time distributions be lattice. Let $\graph$ contain a spanning tree. Then the stationary distribution $\lim_{t\to\infty}p(t)$ exists, is unique, and is given by:
    \begin{align}
        \lim_{t\to\infty}p_i(t)=\pi_i=\frac{\sum_{\mathcal{T}_{[i]}}w_\kappa(\mathcal{T})}{\sum_{l\in V(\graph)}\sum _{\mathcal{T}_{[l]}}w_\kappa(\mathcal{T})}
    \end{align}
    for all $i\in V(\graph)$.
\end{corollary}
\begin{proof}
    This follows immediately from applying Theorem \ref{theorem:SMCTT} and the convergence of $\lim_{t\to\infty}p(t)$ for finite-mean sojourn times and nonlattice transition-time distributions.
\end{proof}

\section{Hitting times}
While the graph-theoretic identities for stationary occupation probabilities of Markov chains are a classical result which has been known for about 60 years~\cite{hill1966studies, schnakenberg1976network}, the extension of the framework to hitting times, also known as first-passage times, has received less attention. In 2007, Chebotarev derived a graph-theoretic identity for the mean hitting times~\cite{chebotarev2007graph}, and it was only recently that the formulas were generalized to higher order moments and hitting-time densities in Laplace space, mainly by Nam and Gunawardena~\cite{nam2023linear, nam2025algebraic, voits2026emergence}.

\new{Graph theoretic techniques to calculate hitting-time distributions of semi-Markov processes have also been established in the flowgraph literature~\cite{huzurbazar2004flowgraph,collins2013multi}, where transforms of transition-time distributions are combined along paths and loops. In contrast, we here show that the} graph-theoretic framework \new{formulated in terms of spanning trees and spanning forests, also extends} to hitting times in semi-Markov networks. To that end, we first define the hitting time to a state $n\in V(\graph)$:
\begin{definition}
    Let $X_t$ be a semi-Markov process on a network on $\graph=\big(\{1,...,N\}, E(\graph)\big)$. For $n\in V(\graph)$, the \textit{hitting time to state} $n$ is a random variable defined as:
    \begin{align}
        T_n:=\inf\{t\geq 0~|~X_t=n\}.
    \end{align}
    We denote the cumulative distribution function of the hitting time as $F_{n}(t):=\mathbb{P}(T_n\leq t)$.
\end{definition}
Without loss of generality, we will set the state of interest to $n=N$ and denote the hitting time as $T:=T_N$ and the cumulative distribution function as $F(t):=F_{N}(t)$. Also, as only the first time that $N$ is reached matters, $N$ can be assumed to be an absorbing state.

We define the Laplace-Stieltjes transform of the hitting-time
distribution by
\begin{align}
    \tilde{f}(s)=\integ{[0,\infty)}{}{e^{-st}F(}{t)}=s\tilde{F}(s).\label{eq:FPT_distribution}
\end{align}
In particular, $f(t)=\tder F(t)$ wherever the derivative exists.
The hitting-time distribution can be related to the occupation probabilities $p(t)$:
\begin{proposition}\label{prop_fpt_distribution_semi_Markov}
    Let $X_t$ be a semi-Markov process on a network on $\graph=\big(\{1,...,N\}, E(\graph)\big)$, where $N$ is an absorbing state, with initial distribution $p_0$. Then
    \begin{align}
        F(t)&=p_N(t)=1-\sum_{j=1}^{N-1}p_j(t),
        \\\tilde{f}(s)&=p_{0,N}+\sum_{j=1}^{N-1}p_{0,j}\frac{\sum_{i=1}^{N-1}\Tilde{K}_{iN}(s)\sum_{\mathcal{F}^{s,j\to i}_{[i,N,\ast]}}w(\mathcal{F})}{\sum_{\mathcal{F}^{s}_{[N,\ast]}}w(\mathcal{F})}\label{eq:FPT_distribution_graph_theory_1},
    \end{align}
    where the sums are taken over the augmented graph $\graph^\ast$ as before.
\end{proposition}
\begin{proof}
    The first identity follows immediately from the fact that $N$ is absorbing, as once $X_t$ reaches $N$, it stays in this state, so $p_N(t)=\mathbb{P}(T\leq t)=F(t)$.

    The Laplace transform of $F(t)$ therefore reads:
    \begin{align}
        \tilde{F}(s)=\frac{1}{s}-\sum_{j=1}^{N-1}\tilde{p}_j(s),
    \end{align}
    and applying Eq.~(\ref{eq:FPT_distribution}) and Eq.~(\ref{eq:transient_semi_MCTT}) then yields:
    \begin{align}
    \Tilde{f}(s)&=s\tilde{F}(s)\\&=1-s\sum_{i=1}^{N-1}\Tilde{p}_i(s)&\\&=1-\sum_{j=1}^Np_{0,j}\frac{s\sum_{i=1}^{N-1}\sum_{\mathcal{F}^{s,j\to i}_{[i,\ast]}}w(\mathcal{F})}{\sum_{\mathcal{T}^{s}_{[\ast]}}w(\mathcal{T})}\\ &=\sum_{j=1}^Np_{0,j}\frac{\sum_{\mathcal{T}^{s}_{[\ast]}}w(\mathcal{T})-s\sum_{i=1}^{N-1}\sum_{\mathcal{F}^{s,j\to i}_{[i,\ast]}}w(\mathcal{F})}{\sum_{\mathcal{T}^{s}_{[\ast]}}w(\mathcal{T})}.
    \end{align}
    Observe that for every $j\in V(\graph)\setminus\{N\}$, the sum over the spanning trees of the augmented graph can be decomposed as follows:
    \begin{align}
        \sum_{\mathcal{T}^{s}_{[\ast]}}w(\mathcal{T})&=s\sum_{i=1}^{N-1}\sum_{\mathcal{F}^{s,j\to i}_{[i,\ast]}}w(\mathcal{F})+s\sum_{\mathcal{F}^{s,j\to N}_{[N,\ast]}}w(\mathcal{F})\\&=s\sum_{i=1}^{N-1}\sum_{\mathcal{F}^{s,j\to i}_{[i,\ast]}}w(\mathcal{F})+s\sum_{i=1}^{N-1}\Tilde{K}_{iN}(s)\sum_{\mathcal{F}^{s,j\to i}_{[i,N,\ast]}}w(\mathcal{F}).
    \end{align}
    Thus, the Laplace transform of the distribution can be expressed as:
    \begin{align}
        \Tilde{f}(s)&=p_{0,N}+\sum_{j=1}^{N-1}p_{0,j}\frac{\sum_{i=1}^{N-1}\Tilde{K}_{iN}(s)\sum_{\mathcal{F}^{s,j\to i}_{[i,N,\ast]}}w(\mathcal{F})s}{\sum_{\mathcal{T}^{s}_{[\ast]}}w(\mathcal{T})}\\ &=p_{0,N}+\sum_{j=1}^{N-1}p_{0,j}\frac{\sum_{i=1}^{N-1}\Tilde{K}_{iN}(s)\sum_{\mathcal{F}^{s,j\to i}_{[i,N,\ast]}}w(\mathcal{F})}{\sum_{\mathcal{F}^{s}_{[N,\ast]}}w(\mathcal{F})},
    \end{align}
    using that the only outgoing edge of $N$ in the augmented graph is $(N,\ast)$, which therefore contributes a factor $s$ to all terms in the sum over the spanning trees.
\end{proof}

Note that $N$ and $\ast$ can now also be contracted into a single vertex. So, instead of considering the augmented graph $\graph^\ast$, we may alternatively define the graph $ \graph^{N^\ast}=(V(\graph^{N^\ast}),E(\graph^{N^\ast}))$, where
\begin{align}
    V(G^{N\ast})&:=\{1,\ldots,N-1,N^\ast\},\\E(G^{N\ast})&:=\{(i,j)\in E(G)~|~ i,j\neq N\}\cup\{(i,N^\ast)~|~ i=1,\ldots,N-1\}.
\end{align}
The weights of the contracted graph for $(i,j)\in E(\graph^{N\ast})$ are defined by
\begin{align}
    w(i,j)=\begin{cases}
        \tilde{K}_{ij}(s)&\text{if } i,j\in V(G^{N^\ast})\setminus\{N^\ast\}\\ s+\tilde {K}_{iN}(s)&\text{if }i\in V(G^{N^\ast})\setminus\{N^\ast\},~j=N^\ast
    \end{cases},
\end{align}
where we set $\tilde{K}_{iN}(s)=0$ if $(i,N)\notin E(G)$.

Summing over the spanning forests of $\graph^{N^\ast}$ instead of $\graph^\ast$, the Laplace transform of the distribution then reads:
\begin{align}
    \Tilde{f}(s)&=p_{0,N}+\sum_{j=1}^{N-1}p_{0,j}\frac{\sum_{i=1}^{N-1}\Tilde{K}_{iN}(s)\sum_{\mathcal{F}^{s,j\to i}_{[i,N^\ast]}}w(\mathcal{F})}{\sum_{\mathcal{T}^{s}_{[N^\ast]}}w(\mathcal{T})}.
\end{align}
This shows that the graph-theoretic expression of the Laplace transform of the hitting-time distribution is again analogous to the result found for Markov networks~\cite{voits2026emergence}, replacing the constant transition rates $k_{ij}$ in the Markovian case by Laplace-transformed memory kernels $\tilde{K}_{ij}(s)$.  

Since the moments of the sojourn times can be infinite, so can the moments of the hitting time. To generalize the graph-theoretic formulas for the moments, we therefore need to restrict the class of semi-Markov networks. For that, we \new{also} need the following \new{technical} lemma:

\begin{lemma}\label{lemma:finite_K}
    Consider a semi-Markov network on the graph $\graph=(\{1,...,N\},E(\graph))$. Let $(i,j)\in E(\graph)$, and
    \begin{align}
        \integ{0}{\infty}{t^nQ_i(}{t})<\infty
    \end{align}
    for an $n\in \mathbb{N}$, i.e., the $n$-th moment of the sojourn time in state $i$ is finite. Then $\lim_{s\to 0}\frac{\text{d}^m}{\text{d}s^m}\tilde{K}_{ij}(s)$ is finite for all $m\in\{0,...,n-1\}$ and:
    \begin{align}
        \lim_{s\to 0}s\frac{\text{d}^n}{\text{d}s^n}\tilde{K}_{ij}(s)=0.
    \end{align}
\end{lemma}
\begin{proof}
    \new{See Appendix \ref{appendix_proof_Lemma_14}}.
\end{proof}

This now allows us to generalize the graph-theoretic formulas for the moments of the hitting time known for Markov networks to semi-Markov networks:

\begin{theorem}\label{theorem_fpt_semi_Markov}
    Let $X_t$ be a semi-Markov process on a network on $\graph=\big(V(\graph), E(\graph)\big)$, where $V(\graph)=\{1,...,N\}$, with initial distribution $p_0$. Let $\graph$ contain at least one spanning tree rooted at $N$, and let $T$ be the hitting time to state $N$. If the sojourn times for all states $i\in V(\graph)\setminus\{N\}$ have finite moments up to order $n\in \mathbb{N}$, then the following holds:
    \begin{align}
        \mathbb{E}[T^k]&=(-1)^{k-1}k\sum_{j=1}^Np_{0,j}\frac{\text{d}^{k-1}}{\text{d}s^{k-1}}\frac{\sum_{i=1}^{N-1}\sum_{\mathcal{F}^{s,j\to i}_{[i,N,\ast]}}w(\mathcal{F})}{\sum_{\mathcal{F}^{s}_{[N,\ast]}}w(\mathcal{F})}\Big|_{s=0}<\infty,\label{eq:moments_fpt}
    \end{align}
    for every $k\in\{1,...,n\}$, and for the mean in particular:
    \begin{align}
        \mathbb{E}[T]=\sum_{j=1}^{N}p_{0,j}\frac{\sum_{i=1}^{N-1}\sum_{\mathcal{F}^{j\to i}_{[i,N]}}w_\kappa(\mathcal{F})}{\sum_{\mathcal{T}_{[N]}}w_\kappa(\mathcal{T})}.\label{eq:MFPT}
    \end{align}
    \begin{proof}
        Using Eq.~(\ref{eq:FPT_distribution_graph_theory_1}), we can express the moments as follows:
        \begin{align}
             \mathbb{E}[T^k]=&(-1)^k\lim_{s\to 0}\frac{\text{d}^k\tilde{f}}{\text{d}s^k}(s)\\=&(-1)^k\lim_{s\to 0}\frac{\text{d}^k}{\text{d}s^k}\Big(1-p_{0,N}-\sum_{j=1}^{N-1}p_{0,j}\frac{s\sum_{i=1}^{N-1}\sum_{\mathcal{F}^{s,j\to i}_{[i,N,\ast]}}w(\mathcal{F})}{\sum_{\mathcal{F}^{s}_{[N,\ast]}}w(\mathcal{F})}\Big)\\ =&(-1)^{k-1}\lim_{s\to 0}\sum_{j=1}^{N-1}p_{0,j}\Big(k\frac{\text{d}^{k-1}}{\text{d}s^{k-1}}\frac{\sum_{i=1}^{N-1}\sum_{\mathcal{F}^{s,j\to i}_{[i,N,\ast]}}w(\mathcal{F})}{\sum_{\mathcal{F}^{s}_{[N,\ast]}}w(\mathcal{F})}\\&+s\frac{\text{d}^{k}}{\text{d}s^k}\frac{\sum_{i=1}^{N-1}\sum_{\mathcal{F}^{s,j\to i}_{[i,N,\ast]}}w(\mathcal{F})}{\sum_{\mathcal{F}^{s}_{[N,\ast]}}w(\mathcal{F})}\Big)\label{eq:FPT_moments_intermediate_step}.
        \end{align}
        The second summand contains terms of at most $s\frac{\text{d}^{n}\tilde{K}_{ij}}{\text{d}s^n}(s)$, which all vanish by Lemma \ref{lemma:finite_K} for $s\to 0$. Since by assumption, there is a spanning tree $\mathcal{T}$ rooted in $N$, and it has weight $w_\kappa(\mathcal{T})>0$ because the sojourn times have finite first moments, i.e., finite means, we have:
        \begin{align}
            \lim_{s\to0}\sum_{\mathcal{F}^{s}_{[N,\ast]}}w(\mathcal{F})=\sum_{\mathcal{T}_{[N]}}w_\kappa(\mathcal{T})>0.
        \end{align}
        Hence, the second summand in Eq.~(\ref{eq:FPT_moments_intermediate_step}) vanishes and 
        \begin{align}
            \mathbb{E}[T^k]&=(-1)^{k-1}k\sum_{j=1}^{N-1}p_{0,j}\frac{\text{d}^{k-1}}{\text{d}s^{k-1}}\frac{\sum_{i=1}^{N-1}\sum_{\mathcal{F}^{s,j\to i}_{[i,N,\ast]}}w(\mathcal{F})}{\sum_{\mathcal{F}^{s}_{[N,\ast]}}w(\mathcal{F})}\Big|_{s=0},
        \end{align}
        which by Lemma \ref{lemma:finite_K} is finite.

        The formula for the mean follows immediately by setting $k=1$:
        \begin{align}
            \mathbb{E}[T]=\sum_{j=1}^{N}p_{0,j}\frac{\sum_{i=1}^{N-1}\sum_{\mathcal{F}^{j\to i}_{[i,N]}}w_\kappa(\mathcal{F})}{\sum_{\mathcal{T}_{[N]}}w_\kappa(\mathcal{T})}.
        \end{align}
    \end{proof}
\end{theorem}

This theorem tells us that the formula for the mean hitting time (Eq.~(\ref{eq:MFPT})) is exactly the same as for a Markov network with the pseudo-rates $\kappa_{ij}$ as transition rates. Semi-Markov corrections appear only in the higher moments. 

\new{For $k=2$, Eq.~(\ref{eq:moments_fpt}) yields
\begin{align}
     \mathbb{E}[T^2]&=-2\sum_{j=1}^Np_{0,j}\frac{\text{d}}{\text{d}s}\frac{\sum_{i=1}^{N-1}\sum_{\mathcal{F}^{s,j\to i}_{[i,N,\ast]}}w(\mathcal{F})}{\sum_{\mathcal{F}^{s}_{[N,\ast]}}w(\mathcal{F})}\Big|_{s=0}\\ &=2\Big(\mathbb{E}[T]\frac{\sum_{l=1}^{N-1}\sum_{\mathcal{F}_{[l,N]}}w_\kappa(\mathcal{F})}{\sum_{\mathcal{T}_{[N]}}w_\kappa(\mathcal{T})}-\frac{\sum_{i,j=1}^{N-1}p_{0,j}\sum_{l=1, l\neq i}^{N-1}\sum_{\mathcal{F}^{j\to i}_{[i,l,N]}}w_\kappa(\mathcal{F})}{\sum_{\mathcal{T}_{[N]}}w_\kappa(\mathcal{T})}\Big)\\ &+2\Big(\mathbb{E}[T]\frac{\sum_{\mathcal{T}_{[N]}}\sum_{(m,n)\in E(\mathcal{T})}(\frac{\text{d}}{\text{d}s}\ln{\tilde{K}_{mn}(s)}\big|_{s=0})w_\kappa(\mathcal{T})}{\sum_{\mathcal{T}_{[N]}}w_\kappa(\mathcal{T})}\\&-\frac{\sum_{j=1}^Np_{0,j}\sum_{i=1}^{N-1}\sum_{\mathcal{F}^{j\to i}_{[i,N]}}\sum_{(m,n)\in E(\mathcal{F})}(\frac{\text{d}}{\text{d}s}\ln{\tilde{K}_{mn}(s)}\big|_{s=0})w_\kappa(\mathcal{F})}{\sum_{\mathcal{T}_{[N]}}w_\kappa(\mathcal{T})}\Big) 
\end{align}
The first summand corresponds to the second moment of the Markovian process with pseudo rates $\kappa_{ij}$ on the network, while the second term can be seen as the semi-Markov correction. Evidently, it vanishes for $\tilde{K}_{nm}(s)=k_{nm}=const.$ and more generally, for 

\begin{align}
    \frac{\text{d}}{\text{d}s}\ln{\tilde{K}_{mn}(s)}\big|_{s=0}&=\frac{\frac{\text{d}}{\text{d}s}\tilde{K}_{mn}(s)}{\tilde{K}_{mn}(s)}\Big|_{s=0}\\ &=\frac{1}{\kappa_{mn}}\Big(\frac{\frac{\text{d}}{\text{d}s}\tilde{Q}_{mn}(s)\big|_{s=0}}{\tilde{\phi}_m(0)}-\frac{\tilde{Q}_{mn}(0)\frac{\text{d}}{\text{d}s}\tilde{\phi}_m(s)\big|_{s=0}}{\tilde{\phi}_m(0)^2}\Big)\\ &= \frac{1}{p_{mn}}\Big(\frac{p_{mn}\mathbb{E}[\tau_m^2]}{2\mathbb{E}[\tau_m]}-\integ{0}{\infty}{t~Q_{mn}(}{t)}\Big)\\ &=\Big(\frac{\mathbb{E}[\tau_m^2]}{2\mathbb{E}[\tau_m]}-\mathbb{E}[\tau_m~|~m\to n]\Big),
\end{align}
where $\mathbb{E}[\tau_m~|~m\to n]$ denotes the expectation value of the sojourn time in state $m$ conditional on the $m\to n$ transition. If the sojourn time in state $m$ is independent of the outgoing transition, as in the Markovian case, we have $\mathbb{E}[\tau_m~|~m\to n]=\mathbb{E}[\tau_m]$, and the condition for the semi-Markov correction to vanish reduces to:
\begin{align}
    0=\frac{\text{d}}{\text{d}s}\ln{\tilde{K}_{mn}(s)}\big|_{s=0}&=\Big(\frac{\mathbb{E}[\tau_m^2]}{2\mathbb{E}[\tau_m]}-\mathbb{E}[\tau_m]\Big)\Rightarrow CV_{\tau_m}^2:=\frac{\text{Var}[\tau_m]}{\mathbb{E}[\tau_m]^2}=1
\end{align}
So, if the coefficient of variation $CV_{\tau_m}$ is equal to $1$ for all $m\in V(\graph)$, we get the same second moment for the semi-Markov process as for a Markov process with effective rates $\kappa_{ij}$.

}

Extending the ideas from \cite{pitman2018tree}, this identical shape allows us to readily prove that Kemeny's constant is also a constant in semi-Markov networks:  
\begin{proposition}[Kemeny's constant for semi-Markov networks]
    Consider a semi-Markov network on $\graph=\big(V(\graph), E(\graph)\big)$, where $V(\graph)=\{1,...,N\}$, with finite-mean sojourn time distributions. Let $\graph$ contain a spanning tree rooted at every $i\in V(\graph)$. Let $T_{i\to j}$ be the hitting time to state $j$ for a process starting at $i$, i.e., $X_0=i$. Then, there exists $K\in(0,\infty)$ (\textit{Kemeny's constant}) such that $\forall i\in V(\graph)$:
    \begin{align}
        \sum_{j=1,j\neq i}^N\pi_j\mathbb{E}[T_{i\to j}]=K,
    \end{align}
    \new{with $(\pi_j)_{j\in V(\graph)}$ as in Eq.~(\ref{eq:stationary_solution}).}
    Moreover, $K$ is given by:
    \begin{align}
        K=\frac{\sum_{1\leq j<l\leq N}\sum_{\mathcal{F}_{[l,j]}}w_\kappa(\mathcal{F})}{\sum_{l=1}^N\sum_{\mathcal{T}_{[l]}}w_\kappa(\mathcal{T})}.
    \end{align}
\end{proposition}
\begin{proof}
    \new{Using the expressions in Eq.~(\ref{eq:stationary_solution}) and Eq.~(\ref{eq:MFPT}), we obtain:}
    \begin{align}
        \sum_{j=1,j\neq i}^N\pi_j\mathbb{E}[T_{i\to j}]&=\frac{\sum_{j=1,j\neq i}^N\sum_{l=1}^{N}\sum_{\mathcal{F}^{i\to l}_{[l,j]}}w_\kappa(\mathcal{F})}{\sum_{l=1}^N\sum_{\mathcal{T}_{[l]}}w_\kappa(\mathcal{T})}\\&=\frac{\sum_{j=1}^N\sum_{l=1}^{N}\sum_{\mathcal{F}^{i\to l}_{[l,j]}}w_\kappa(\mathcal{F})}{\sum_{l=1}^N\sum_{\mathcal{T}_{[l]}}w_\kappa(\mathcal{T})}\\&=\frac{\sum_{1\leq j<l\leq N}\sum_{\mathcal{F}_{[l,j]}}w_\kappa(\mathcal{F})}{\sum_{l=1}^N\sum_{\mathcal{T}_{[l]}}w_\kappa(\mathcal{T})},
    \end{align}
    where in the second line, we used that $\sum_{\mathcal{F}^{i\to l}_{[l,i]}}w_\kappa (\mathcal{F})=0$ for all $l$, since $i$ cannot be a root and connected to $l\neq i$ at the same time, and the third line follows because the sum is taken over all forests with two roots and $i$ is always connected to exactly one of them. This makes the expression independent of $i$, and we conclude that:
    \begin{align}
        K=\frac{\sum_{1\leq j<l\leq N}\sum_{\mathcal{F}_{[l,j]}}w_\kappa(\mathcal{F})}{\sum_{l=1}^N\sum_{\mathcal{T}_{[l]}}w_\kappa(\mathcal{T})}.
    \end{align}
\end{proof}

An algebraic proof of the constancy of Kemeny's constant in semi-Markov networks can be found in \cite{hunter2018kemeny}.

Lastly, Eq.~(\ref{eq:FPT_distribution_graph_theory_1}) also allows us to study the shape of the hitting-time tail in the presence of infinite-mean sojourn times. To this end, we again need to restrict the class of possible sojourn time distributions to specific small-$s$ regimes as in Assumption \ref{assumption:small_s_sojourn_times}.
In this case, we also need a more restrictive notion than $\alpha^\ast$, because only the vertices reachable from the initial distribution $p_0$ contribute and also because $N$ was chosen to be an absorbing vertex. We therefore define:
\begin{align}
    \alpha_{p_0}(\graph):=\min_{i=1,...,N-1}\{\alpha_i~|~\exists j\in V(\graph):p_{0,j}>0\text{ and }j \text{ is connected to }i \text{ in }\graph\}.
\end{align}
Note that the definition excludes $N$ because we would always have $\alpha_{p_0}(\graph)=0$ if $N$ were included. Then the following asymptotic result holds:
\begin{theorem}\label{Theorem:FPT_asymptotics}
    Let $X_{t\geq 0}$ be a semi-Markov process on a network $\graph$ with at least one spanning tree with root $N\in V(\graph)$ with initial condition $p_0$ such that $p_{0,N}<1$ and sojourn times that satisfy Assumption \ref{assumption:small_s_sojourn_times}. Then the small-$s$ behavior of the hitting-time distribution in Laplace space is given by:
    \begin{align}
     \tilde{f}(s)=1-\sum_{j=1}^Np_{0,j}\frac{\sum_{i:\alpha_i=\alpha_{p_0}(\graph)}\sum_{\mathcal{F}^{ j\to i}_{[i,N]}}w_\kappa(\mathcal{F})}{\sum_{\mathcal{T}_{[N]}}w_\kappa(\mathcal{T})}s^{\alpha_{p_0}(\graph)}+o(s^{\alpha_{p_0}(\graph)})\label{eq:FPT_distribution_Laplace_asymptotics}.
     \end{align}
    If $\alpha_{p_0}(\graph)<1$, this yields the following asymptotic scaling in time domain for the survival probability:
    \begin{align}
        \mathbb{P}(T>t)\sim \frac{\sum_{j=1}^Np_{0,j}\sum_{i:\alpha_i=\alpha_{p_0}(\graph)}\sum_{\mathcal{F}^{ j\to i}_{[i,N]}}w_\kappa(\mathcal{F})}{\Gamma(1-\alpha_{p_0}(\graph))\sum_{\mathcal{T}_{[N]}}w_\kappa(\mathcal{T})}t^{-\alpha_{p_0}(\graph)}.
    \end{align}  
\end{theorem}
\begin{proof}
    We again use that Assumption \ref{assumption:small_s_sojourn_times} yields the following small-$s$ scaling for a spanning forest $\mathcal{F}$ of $\graph$ with roots $\{i_1,...,i_n\}\subseteq V(\graph)$: 
    \begin{align}
    w(\mathcal{F})=w_\kappa(\mathcal{F})s^{|E(\mathcal{F})|-\sum_{k\in V(\graph)}\alpha_k+\sum_{k=1}^n\alpha_{i_k}}+o(s^{|E(\mathcal{F})|-\sum_{k\in V(\graph)}\alpha_k+\sum_{k=1}^n\alpha_{i_k}}).
    \end{align}
    Since $N$ is an absorbing state, we have:
    \begin{align}
        \sum_{\mathcal{F}^{s,j\to i}_{[i,N,\ast]}}w(\mathcal{F})&=\sum_{\mathcal{F}^{s,j\to i}_{[i,N]}}w_\kappa(\mathcal{F})s^{N-2-\sum_{k\in V(\graph)}\alpha_k+\alpha_{i}}+o(s^{N-2-\sum_{k\in V(\graph)}\alpha_k+\alpha_{i}}),\\\sum_{\mathcal{F}^{s}_{[N,\ast]}}w(\mathcal{F})&=\sum _{\mathcal{T}_{[N]}}w_\kappa(\mathcal{T})s^{N-1-\sum_{k\in V(\graph)}\alpha_k}+o(s^{N-1-\sum_{k\in V(\graph)}\alpha_k}).
    \end{align}
    Hence, using Eq.~(\ref{eq:FPT_distribution_graph_theory_1}), we find the following small-$s$ regime for $\tilde{f}(s)$:
    \begin{align}
     \tilde{f}(s)&=1-s\sum_{j=1}^Np_{0,j}\frac{\sum_{i=1}^{N-1}\sum_{\mathcal{F}^{s,j\to i}_{[i,N,\ast]}}w(\mathcal{F})}{\sum_{\mathcal{F}^{s}_{[N,\ast]}}w(\mathcal{F})}\\ &=1-\sum_{j=1}^Np_{0,j}\frac{\sum_{i:\alpha_i=\alpha_{p_0}(\graph)}\sum_{\mathcal{F}^{ j\to i}_{[i,N]}}w_\kappa(\mathcal{F})}{\sum_{\mathcal{T}_{[N]}}w_\kappa(\mathcal{T})}s^{\alpha_{p_0}(\graph)}+o(s^{\alpha_{p_0}(\graph)}).
     \end{align}
    For the result in time domain, we observe that:
    \begin{align}
        \integ{0}{\infty}{\mathbb{P}(T>t)e^{-st}}{t}&=\integ{0}{\infty}{(1-F(t))e^{-st}}{t}\\&=\frac{1-\tilde{f}(s)}{s}\\ &=\sum_{j=1}^Np_{0,j}\frac{\sum_{i:\alpha_i=\alpha_{p_0}(\graph)}\sum_{\mathcal{F}^{ j\to i}_{[i,N]}}w_\kappa(\mathcal{F})}{\sum_{\mathcal{T}_{[N]}}w_\kappa(\mathcal{T})}s^{\alpha_{p_0}(\graph)-1}+o(s^{\alpha_{p_0}(\graph)-1}).
    \end{align}
    The asymptotic behavior of $\mathbb{P}(T>t)$ follows from Karamata's Tauberian theorem for Laplace-Stieltjes transforms~\cite{feller1991introduction}, for the case $\alpha_{p_0}(\graph)<1$:
    \begin{align}
        \mathbb{P}(T>t)\sim \frac{\sum_{j=1}^Np_{0,j}\sum_{i:\alpha_i=\alpha_{p_0}(\graph)}\sum_{\mathcal{F}^{ j\to i}_{[i,N]}}w_\kappa(\mathcal{F})}{\Gamma(1-\alpha_{p_0}(\graph))\sum_{\mathcal{T}_{[N]}}w_\kappa(\mathcal{T})}t^{-\alpha_{p_0}(\graph)}.
    \end{align}
\end{proof}

The previous theorem has an intuitive explanation: in the presence of infinite-mean sojourn times, the tail of the hitting-time distribution is dominated by the slowest sojourn time(s). Note that in the case $\alpha_{p_0}(\graph)=1$, Eq.~(\ref{eq:FPT_distribution_Laplace_asymptotics}) is not sufficient to characterize the asymptotic behavior as it gives only the expansion of the Laplace–Stieltjes transform to first order.

\section{Weak ergodicity breaking}\label{sec:weak_ergodicity_breaking}
In the last section, we studied the graph-theoretic structure of hitting times in semi-Markov networks, which, from a physics perspective, describe the time the process takes to reach a certain target state. Time-additive observables constitute another class of observables which are often central in statistical physics and related fields, e.g., in the theory of stochastic thermodynamics~\cite{seifert2025stochastic}. In the following, we apply the graph-theoretic formalism to time-additive observables, which we first define:
\begin{definition}
    Let $X_t$ be a semi-Markov process on a network $\graph=(V(\graph), E(\graph))$ with the underlying discrete-time Markov process $(X^d_n)_{n\in\mathbb{N}_0}$. We define the \textit{number of transitions along an edge }$(i,j)\in E(\graph)$ as:
    \begin{align}
        N_{ij}(t):=\big|\{n\in\mathbb{N}~|~t_n\leq t,~{X^d_{n-1}=i\land X^d_n=j} \}\big|.
    \end{align}
    A \textit{time-additive observable} is a random variable of the shape:
    \begin{align}
    A(t)=\integ{0}{t}{g(X_{t'})}{t'}+\sum_{(i,j)\in E(\graph)}h(i,j)N_{ij}(t),
    \end{align}
    where 
    \begin{align*}
        g:V(\graph)&\to \mathbb{R}\\
        i&\mapsto g_i
    \end{align*}
    defines the \textit{state weights} and
    \begin{align*}
        h:E(\graph)&\to \mathbb{R}\\
        (i,j)&\mapsto h_{ij}
    \end{align*}
    defines the \textit{edge weights} of $A(t)$. If $h=0$, we call $A(t)$ a \textit{time-additive state observable}, and if $g=0$, we call $A(t)$ a \textit{time-additive transition observable}. 
\end{definition}

There are several important examples of time-additive observables, including the following:
\begin{itemize}
    \item $g_i=\delta_{in},~ h=0$ describes the \textit{occupation time of state} $n\in V(\graph)$, i.e., the time the process has spent in state $n$ up to time $t$.
    \item $g=0,~h_{ij}=1$ for all $(i,j)\in E(\graph)$ is the total jump count of the process.
    \item For a Markov process with transition rates $(k_{ij})_{(i,j)\in E(\graph)}$ and bidirectional transitions, $g=0,~h_{ij}=\log\frac{k_{ij}}{k_{ji}}$ defines the medium entropy production.
\end{itemize}

An object of central interest will be the moment generating function of a time-additive observable, which we define as
\begin{align}
    G_i(t,\lambda):=\mathbb{E}[ e^{\lambda A(t)}~|~X_0=i].
\end{align}
In the Markovian case, 
$G_i(t,\lambda)$ satisfies the usual Feynman-Kac equation generated by a tilted generator~\cite{fitzsimmons1999kac}, while in the semi-Markov case, one obtains a renewal equation, which may be viewed as a generalization of it. \new{Such renewal-equations constitute a standard approach to time-additive observables in the theory of semi-Markov processes \cite{cinlar1975introduction,janssen2006applied}}:
\begin{proposition}\label{prop:generalized_Feynman_Kac}
     Let $A(t)$ be a time-additive observable of a semi-Markov network on $\graph=(V(\graph), E(\graph))$. Define $\phi_i^{(\lambda)}(t):=\phi_i(t)e^{\lambda g_it}$ for $i\in V(\graph)$ and $Q_{ij}^{(\lambda)}(\text{d}t):=e^{\lambda\left(g_i t+h_{i,j}\right)}Q_{ij}(\text{d}t)$ for $(i,j)\in E(\graph)$. Then, for $\tilde{G}(\lambda,s)=(\tilde{G}_i(\lambda,s))_{i\in V(\graph)}$ and every $ s>0$, there exists $\epsilon_s>0$ such that for every $\lambda\in (-\epsilon_s,\epsilon_s)$, the following identity holds:
     \begin{align}
         \big(\mathbb{1}-\tilde{Q}^{(\lambda)}(s)\big)\tilde{G}(\lambda,s)=\tilde{\phi}^{(\lambda)}(s).\label{eq:generalized_Feynmann_Kac}
     \end{align}
\end{proposition}
\begin{proof}
    \new{See Appendix \ref{appendix_section_weak_ergodicity_breaking}.}
\end{proof}

An expression for the $k$-th moment of $A(t)$ can be obtained by evaluating the $k$-th derivative of $G$ at $\lambda=0$. As the involved terms get increasingly complicated, we here only give the result for the mean:
\begin{proposition}\label{Prop:time_additive_observable_mean}
    Let $A(t)$ be a time-additive observable of a semi-Markov network on $\graph=(V(\graph), E(\graph))$. Then, the Laplace transform of the expectation value of $A(t)$ is given by:
    \begin{align}
        \widetilde{\mathbb{E}[A]}(s)=\frac{1}{s}\sum_{j\in V(\graph)}\tilde{p}_{j}(s)\big(\sum_{l:(j,l)\in E(\graph)}h_{jl}\tilde{K}_{jl}(s)+g_j\big),
    \end{align}
    and whenever the inverse Laplace transforms $K_{ij}$ exist as ordinary
    functions, the expression in Laplace space corresponds to
    \begin{align}
        \mathbb{E}[A(t)]=\integ{0}{t}{\big(\sum_{i\in V(\graph)}g_ip_i(t')+\sum_{(i,j)\in E(\graph)}h_{ij}\integ{0}{t'}{K_{ij}(t'-t'')p_i(t'')}{t''}\big)}{t'}.
    \end{align}
\end{proposition}
\begin{proof}
    We consider the Laplace transform of the moment generating function $\tilde{G}_i(\lambda,s)$ of $A(t)$ for a process starting in state $i\in V(\graph)$ to calculate the first moment as its first derivative at $\lambda=0$.    
    Applying the derivative with respect to $\lambda$ to Eq.~(\ref{eq:generalized_Feynmann_Kac}) yields:
    \begin{align}
        &\big(-\frac{\partial}{\partial\lambda}\tilde{Q}^{(\lambda)}(s)\big)\tilde{G}(\lambda,s)+\big(\mathbb{1}-\tilde{Q}^{(\lambda)}(s)\big)\frac{\partial}{\partial\lambda}\tilde{G}(\lambda,s)=\frac{\partial}{\partial\lambda}\tilde{\phi}^{(\lambda)}(s)\\ \Rightarrow&\big(\mathbb{1}-\tilde{Q}^{(\lambda)}(s)\big)\frac{\partial}{\partial\lambda}\tilde{G}(\lambda,s)=\frac{\partial}{\partial\lambda}\tilde{\phi}^{(\lambda)}(s)+\big(\frac{\partial}{\partial\lambda}\tilde{Q}^{(\lambda)}(s)\big)\tilde{G}(\lambda,s).
    \end{align}
    Next, we evaluate this expression at $\lambda = 0$, using that $\newnew{\tilde{G}_i(0,s)}=\frac{1}{s}$ for all $i\in V(\graph)$ and express $(\mathbb{1}-\tilde{Q}^{(\lambda)}(s))\big|_{\lambda=0}=(\mathbb{1}-\tilde{Q}(s))$ in terms of $\tilde{W}(s)$ according to Lemma \ref{lemma:Q_W_identity} to get:
    \begin{align}
        \diag{\tilde{\phi}(s)}\big(s\mathbb{1}-\tilde{W}^T(s)\big)\frac{\partial}{\partial\lambda}\tilde{G}(\lambda,s)\Big|_{\lambda=0}=\frac{\partial}{\partial\lambda}\big(\tilde{\phi}^{(\lambda)}(s)+\frac{1}{s}\tilde{Q}^{(\lambda)}(s)\mathbf{1}\big)\Big|_{\lambda=0},
    \end{align}
    where $\mathbf{1}_i=1$ for all $i\in V(\graph)$. For $s>0$, $\diag{\tilde{\phi}(s)}\big(s\mathbb{1}-\tilde{W}^T(s)\big)$ can be inverted according to Proposition \ref{Prop:inverse_of_s-W}:
    \begin{align}
        \frac{\partial}{\partial\lambda}\tilde{G}(\lambda,s)\Big|_{\lambda=0}=\big(s\mathbb{1}-\tilde{W}^T(s)\big)^{-1}\diag{\tilde{\phi}(s)}^{-1}\frac{\partial}{\partial\lambda}\big(\tilde{\phi}^{(\lambda)}(s)+\frac{1}{s}\tilde{Q}^{(\lambda)}(s)\mathbf{1}\big)\Big|_{\lambda=0}.
    \end{align}
    Evaluating the terms on the right-hand side at $\lambda=0$, we find:  
    \begin{align}
        \frac{\partial}{\partial\lambda}\tilde{Q}^{(\lambda)}_{ij}(s)\Big|_{\lambda=0}&=h_{ij}\Tilde{Q}_{ij}(s)-g_i\frac{\text{d}}{\text{d}s}\Tilde{Q}_{ij}(s),\\ \frac{\partial}{\partial\lambda}\tilde{\phi}_i^{(\lambda)}(s)\Big|_{\lambda=0}&=g_i\integ{0}{\infty}{\phi_i(t)te^{-st}}{t}\\&=-g_i\frac{\text{d}}{\text{d}s}\tilde{\phi}_i(s)\\ &=-g_i\frac{\text{d}}{\text{d}s}\frac{1-\Tilde{Q}_i(s)}{s}\\ &=g_i\big(\frac{\tilde{\phi}_i(s)}{s}+\frac{1}{s}\frac{\text{d}}{\text{d}s}\Tilde{Q}_i(s)\big).
    \end{align}
    Therefore,
    \begin{align}
        \frac{\partial}{\partial\lambda}\big(\tilde{\phi}^{(\lambda)}(s)+\frac{1}{s}\tilde{Q}^{(\lambda)}(s)\mathbf{1}\big)_i\big|_{\lambda=0}=\frac{1}{s}\big(g_i\tilde{\phi}_i(s)+\sum_jh_{ij}\Tilde{Q}_{ij}(s)\big),
    \end{align}
    and we can apply Proposition \ref{Prop:inverse_of_s-W} to express the inverse 
    \begin{align}
        \partial_\lambda\tilde{G}_i^{(\lambda)}(s)\Big|_{\lambda=0}&=\frac{1}{s}\sum_j\big(s\mathbb{1}-\tilde{W}(s)\big)^{-1}_{ji}\big(\frac{\sum_{l:(j,l)\in E(\graph)}h_{jl}\Tilde{Q}_{jl}(s)+g_j\tilde{\phi}_j(s)}{\tilde{\phi}_j(s)}\big)\\ &=\frac{1}{s}\sum_{j\in V(\graph)}\frac{\sum_{\mathcal{F}^{s,i\to j}_{[j,\ast]}}w(\mathcal{F})}{\sum_{\mathcal{T}^{s}_{[\ast]}}w(\mathcal{T})}\big(\sum_{l:(j,l)\in E(\graph)}h_{jl}\tilde{K}_{jl}(s)+g_j\big).
    \end{align}
    The expectation value of $A(t)$ is given by:
    \begin{align}
        \mathbb{E}[A(t)]=\sum_{i\in V(\graph)}p_{0,i}\frac{\partial}{\partial \lambda}G_i(t,\lambda)\Big|_{\lambda=0},
    \end{align}
    which in Laplace space reads:
    \begin{align}
        \widetilde{\mathbb{E}[A]}(s)=\frac{1}{s}\sum_{i,j\in V(\graph)}p_{0,i}\frac{\sum_{\mathcal{F}^{s,i\to j}_{[j,\ast]}}w(\mathcal{F})}{\sum_{\mathcal{T}^{s}_{[\ast]}}w(\mathcal{T})}\big(\sum_{l:(j,l)\in E(\graph)}h_{jl}\tilde{K}_{jl}(s)+g_j\big).
    \end{align}
    Comparison to Eq.~(\ref{eq:transient_semi_MCTT}) allows us to identify the Laplace transforms of the occupation probabilities in this expression:
    \begin{align}
        \widetilde{\mathbb{E}[A]}(s)=\frac{1}{s}\sum_{j\in V(\graph)}\tilde{p}_{j}(s)\big(\sum_{l:(j,l)\in E(\graph)}h_{jl}\tilde{K}_{jl}(s)+g_j\big),
    \end{align}
    which in time domain yields
    \begin{align}
        \mathbb{E}[A(t)]=\integ{0}{t}{\big(\sum_{i\in V(\graph)}g_ip_i(t')+\sum_{(i,j)\in E(\graph)}h_{ij}\integ{0}{t'}{K_{ij}(t'-t'')p_i(t'')}{t''}\big)}{t'},
    \end{align}
    if the inverse Laplace transforms $K_{ij}$ exist as ordinary functions.
\end{proof}

Finally, we turn our attention to the question of ergodicity, which is the property of a dynamical system that time averages along a single trajectory converge to ensemble averages. More precisely, a time-additive state observable $A(t)$ on a semi-Markov network with a unique \new{$\pi=\lim_{s\to 0}s\tilde{p}_(s)$} is called \textit{ergodic}, if:
\begin{align}
    \frac{A(t)}{t}~\overset{t\to \infty}{\longrightarrow}~\sum_{i\in V(\graph)}g_i \pi_i~~\text{a.s.}
\end{align}

\new{In the following, we consider semi-Markov networks for which the limits of the occupation probabilities in Laplace space exist and are independent of the initial condition, as for example under the conditions of Theorem \ref{theorem:SMCTT}. As shown below, this is sufficient to guarantee convergence of the ensemble-averaged time observables, even if the occupation probabilities themselves do not converge pointwise in time.} 
\new{The goal of this section is to use the graph theoretic framework to} show that weak ergodicity breaking occurs for finite semi-Markov networks with heavy-tailed sojourn times also in the case of generally heterogeneous memory kernels if $g$ is not constant on the support of $\pi$. First, we need to formally introduce the concept of weak ergodicity breaking:

\begin{definition}
    Let $A(t)$ be a time-additive state observable of a semi-Markov network on $\graph=(V(\graph), E(\graph))$ with unique limiting distribution \new{in Laplace space, $\pi=\lim_{s\to 0}s\tilde{p}(s)$}. We say that $A(t)$ exhibits \textit{weak ergodicity breaking} if:
    \begin{align}
        \lim_{t\to\infty}\frac{\mathbb{E}[A(t)]}{t}=\sum_{i\in V(\graph)}g_i \pi_i\label{eq:WEB_mean_condition}
    \end{align}
    but:
    \begin{align}
        \frac{A(t)}{t}~\overset{t\to \infty}{\centernot{\longrightarrow}}~\sum_{i\in V(\graph)}g_i \pi_i~~\text{a.s.}\label{eq:WEB_as_condition}
    \end{align}
\end{definition}

In other words, $A(t)$ exhibits weak ergodicity breaking if $\frac{A(t)}{t}$ remains non-deterministic in the $t\to \infty$-limit.

Eq.~(\ref{eq:WEB_mean_condition}) is always true for a time-additive observable if the limiting state distribution is unique as we shall see next:
\begin{proposition}\label{prop:time_average_mean_converges}
     Let $A(t)$ be a time-additive state observable of a semi-Markov network on $\graph=(V(\graph), E(\graph))$ with unique \new{$\pi=\lim_{s\to 0}s\tilde{p}(s)$}. Then
    \begin{align}
        \lim_{t\to\infty}\frac{\mathbb{E}[A(t)]}{t}=\sum_{i\in V(\graph)}g_i \pi_i.\label{eq:convergence_mean}
    \end{align}
\end{proposition}
\begin{proof}
    By Proposition \ref{Prop:time_additive_observable_mean}, the mean of $A(t)$ is given by:
    \begin{align}
        \mathbb{E}[A(t)]=\sum_{i\in V(\graph)}g_i\integ{0}{t}{p_i(t')}{t'},
    \end{align}
    so it suffices to show that $\frac{1}{t}\integ{0}{t}{p_i(t')}{t'}\overset{t\to\infty}{\longrightarrow}\pi_i$. 
    \new{Since $\integ{0}{t}{p_i(t')}{t'}+t$ is monotonically increasing in $t$, we can use Karamata's Tauberian theorem (see Appendix \ref{appendix_Laplace_Stieltjes}, Theorem \ref{Theorem:Karamata}) to obtain its asymptotical scaling from the small-$s$ scaling in Laplace space. The Laplace transform of $\integ{0}{t}{p_i(t')}{t'}$ reads:
    \begin{align}
        \frac{1}{s}\tilde{p}_i(s)+\frac{1}{s^2}=\frac{\pi_i+1}{s^2}+o(s^{-2}),~\Rightarrow~\integ{0}{t}{p_i(t')}{t'}+t\sim (\pi_i+1) t
    \end{align}
    using that $\lim_{s\to 0}s\tilde{p}(s)=\pi$.
    
    Hence}, we conclude that:
    \begin{align}
        \lim_{t\to\infty}\frac{1}{t}\integ{0}{t}{p_i(t')}{t'}=\pi_i.
    \end{align}
\end{proof}

So, if the limiting state distribution is unique, ergodicity can only be weakly broken depending on whether the second condition (Eq.~(\ref{eq:WEB_as_condition})) holds or not. 

\new{To show the} final theorem \new{characterizing} weak ergodicity breaking in semi-Markov networks, we \new{need} the following lemma:

\begin{lemma}\label{lemma:memory_kernels_small_s}
    Let $\big(\tilde{K}_{ij}(s)\big)_{(i,j)\in E(\graph)}$ be the memory kernels of a semi-Markov network on $\graph =(V(\graph), E(\graph))$. Then
    \begin{align}
        \lim_{s\to 0}\big(\frac{\text{d}}{\text{d}\ln s}\ln\tilde{K}_{ij}(s)-\frac{\text{d}}{\text{d}\ln s}\ln\tilde{K}_{il}(s)\big)=0,
    \end{align}
    for all $i\in V(\graph)$ and $(i,j),(i,l)\in E(\graph)$. If the limit $\lim_{s\to 0}\frac{\text{d}}{\text{d}\ln s}\ln\tilde{K}_{ij}(s)$ exists, it is independent of $j$ and given by:
    \begin{align}
        \lim_{s\to 0}\frac{\text{d}}{\text{d}\ln s}\ln\tilde{K}_{ij}(s)=1+\lim_{s\to 0}\frac{\frac{\text{d}}{\text{d}s}\Tilde{Q}_i(s)}{\tilde{\phi}_i(s)}=:1-\alpha_i\in [0,1].
    \end{align}
    Moreover, if the limit exists for all $i\in V(\graph)$ and the network has a unique $\pi=\lim_{s\to 0}s\tilde{p}(s)$, then the following holds for all $n,m\in V(\graph)$:
    \begin{align}
        \pi_n,\pi_m>0\Rightarrow \alpha_n=\alpha_m.
    \end{align}
\end{lemma}
\begin{proof}
    \new{See Appendix \ref{appendix_section_weak_ergodicity_breaking}}
\end{proof}

In the special case that Assumption \ref{assumption:small_s_sojourn_times} holds, $\alpha_i$ coincides with the previously defined exponent in the small-$s$ expansion of $\Tilde{Q}_i(s)$, but the requirement for Lemma \ref{lemma:memory_kernels_small_s} is less restrictive. 

We now come to the main theorem regarding weak ergodicity breaking:

\begin{theorem}\label{theorem:weak_ergodicity_breaking}
    Let $A(t)$ be a time-additive state observable on a semi-Markov network on $\graph=(V(\graph),E(\graph))$ with a unique \new{$\pi=\lim_{s\to 0}s\tilde{p}(s)$} for all initial conditions. Let the limit $\lim_{s\to 0}\frac{\text{d}}{\text{d}\ln s}\ln\tilde{K}_{ij}(s)=1-\alpha_i$ exist for all $(i,j)\in E(\graph)$. Then:
    \begin{align}
        \lim_{t\to \infty}\text{Var}\big(\frac{A(t)}{t}\big)&=\big(1-\alpha\big)\Big(\sum_{j\in V(\graph)}\pi_jg_j^2-\big(\sum_{j\in V(\graph)}\pi_jg_j\big)^2\Big),\label{eq:convergence_variance}
    \end{align}
    where $\alpha =\alpha_j$ for all $j\in V(\graph)$ such that $\pi_j>0$.
    Moreover, $A(t)$ exhibits weak ergodicity breaking if there is at least one state $n\in V(\graph)$ with $\pi_n>0$ and $\alpha_n<1$ and $ m\in V(\graph)$ with $\pi_n,\pi_m>0$ such that $g_n\neq g_m$.
\end{theorem}
\begin{proof}
    If the limiting distribution is concentrated on an absorbing state, the statement is trivial. In the following, we therefore assume that the states are not absorbing. 

   Consider the long-time limit of the variance of the time average of $A(t)$, given by:
    \begin{align}
        \text{Var}\big(\frac{A(t)}{t}\big)=&\mathbb{E}\Big[\Big|\frac{A(t)}{t}-\mathbb{E}\Big[\frac{A(t)}{t}\Big]\Big|^2\Big]\\=&\frac{1}{t^2}\big(\mathbb{E}[A(t)^2]-\mathbb{E}[A(t)]^2\big),
    \end{align}
    so $\text{Var}\big(\frac{A(t)}{t}\big)\overset{t\to\infty}{\longrightarrow}0$ means that $\frac{A(t)}{t}$ converges to its mean in $L^2$. Since $\big|\frac{A(t)}{t}\big|\leq \max_{i\in V(\graph)}\{|g_i|\}$, i.e., $\frac{A(t)}{t}$ is uniformly bounded, almost-sure convergence to a constant would imply $L^2$-convergence by dominated convergence. Hence a strictly positive limiting variance rules out almost-sure, and indeed, convergence in probability to a constant, so we can prove weak ergodicity breaking from showing that the limit is different from zero.

    We calculate the second moment from the second $\lambda$-derivative of $\tilde{G}(\lambda,s)$ at $\lambda=0$. 
    Taking the second derivative with respect to $\lambda$ of Eq.~(\ref{eq:generalized_Feynmann_Kac}) yields:
    \begin{align}
        &\big(-\frac{\partial^2}{\partial\lambda^2}\tilde{Q}^{(\lambda)}(s)\big)\tilde{G}(\lambda,s)-2\big(\frac{\partial}{\partial\lambda}\tilde{Q}^{(\lambda)}(s)\big)\big(\frac{\partial}{\partial\lambda}\tilde{G}(\lambda,s)\big)+\big(\mathbb{1}-\tilde{Q}^{(\lambda)}(s)\big)\frac{\partial^2}{\partial\lambda^2}\tilde{G}(\lambda,s)=\frac{\partial^2}{\partial\lambda^2}\tilde{\phi}^{(\lambda)}(s)\\ \Rightarrow&\big(\mathbb{1}-\tilde{Q}^{(\lambda)}(s)\big)\frac{\partial^2}{\partial\lambda^2}\tilde{G}(\lambda,s)=\frac{\partial^2}{\partial\lambda^2}\tilde{\phi}^{(\lambda)}(s)+\big(\frac{\partial^2}{\partial\lambda^2}\tilde{Q}^{(\lambda)}(s)\big)\tilde{G}(\lambda,s)+2\big(\frac{\partial}{\partial\lambda}\tilde{Q}^{(\lambda)}(s)\big)\big(\frac{\partial}{\partial\lambda}\tilde{G}(\lambda,s)\big)    
    \end{align}
    Again, we evaluate this expression at $\lambda = 0$, using that $\newnew{\tilde{G}_i(0,s)}=\frac{1}{s}$ for all $i\in V(\graph)$ and express $(\mathbb{1}-\tilde{Q}^{(\lambda)}(s))\big|_{\lambda=0}=(\mathbb{1}-\tilde{Q}(s))$ in terms of $\tilde{W}(s)$:
    \begin{align}
        \diag{\tilde{\phi}(s)}\big(s\mathbb{1}-\tilde{W}^T(s)\big)\frac{\partial^2}{\partial\lambda^2}\tilde{G}(\lambda,s)\Big|_{\lambda=0}=&\frac{\partial^2}{\partial\lambda^2}\big(\tilde{\phi}^{(\lambda)}(s)+\frac{1}{s}\tilde{Q}^{(\lambda)}(s)\mathbf{1}\big)\Big|_{\lambda=0}\\&+2\big(\frac{\partial}{\partial\lambda}\tilde{Q}^{(\lambda)}(s)\big)\big(\frac{\partial}{\partial\lambda}\tilde{G}(\lambda,s)\big)\Big|_{\lambda=0}.
    \end{align}
    For $s>0$, $\diag{\tilde{\phi}(s)}\big(s\mathbb{1}-\tilde{W}^T(s)\big)$ can be inverted according to Proposition \ref{Prop:inverse_of_s-W}:
     \begin{align}
        \frac{\partial^2}{\partial\lambda^2}\tilde{G}(\lambda,s)\Big|_{\lambda=0}=&\big(s\mathbb{1}-\tilde{W}^T(s)\big)^{-1}\diag{\tilde{\phi}(s)}^{-1}\Big(\frac{\partial^2}{\partial\lambda^2}\big(\tilde{\phi}^{(\lambda)}(s)+\frac{1}{s}\tilde{Q}^{(\lambda)}(s)\mathbf{1}\big)\\&+2\big(\frac{\partial}{\partial\lambda}\tilde{Q}^{(\lambda)}(s)\big)\big(\frac{\partial}{\partial\lambda}\tilde{G}(\lambda,s)\big)\Big)\Big|_{\lambda=0}.
    \end{align}
    We can evaluate the right-hand side using that:
    \begin{align}
        \frac{\partial}{\partial\lambda}\tilde{Q}^{(\lambda)}_{ij}(s)\Big|_{\lambda=0}&=-g_i\frac{\text{d}}{\text{d}s}\Tilde{Q}_{ij}(s),\\\frac{\partial^2}{\partial\lambda^2}\tilde{Q}^{(\lambda)}_{ij}(s)\Big|_{\lambda=0}&=g_i^2\frac{\text{d}^2}{\text{d}s^2}\Tilde{Q}_{ij}(s),\\ \frac{\partial^2}{\partial\lambda^2}\tilde{\phi}_i^{(\lambda)}(s)\Big|_{\lambda=0}&=g_i^2\frac{\text{d}^2}{\text{d}s^2}\tilde{\phi}_i(s)\\ &=g_i^2\frac{\text{d}^2}{\text{d}s^2}\frac{1-\Tilde{Q}_i(s)}{s}\\ &=g_i^2\big(2\frac{\tilde{\phi}_i(s)}{s^2}+\frac{2}{s^2}\frac{\text{d}}{\text{d}s}\Tilde{Q}_i(s)-\frac{1}{s}\frac{\text{d}^2}{\text{d}s^2}\Tilde{Q}_i(s)\big),
    \end{align}
    so:
    \begin{align}
        \frac{\partial^2}{\partial\lambda^2}\big(\tilde{\phi}^{(\lambda)}(s)+\frac{1}{s}\tilde{Q}^{(\lambda)}(s)\mathbf{1}\big)_{i}\Big|_{\lambda=0}&=\frac{2g_i^2}{s^2}\big(\tilde{\phi}_i(s)+\frac{\text{d}}{\text{d}s}\Tilde{Q}_i(s)\big)\\&=-\frac{2g_i^2}{s}\frac{\text{d}}{\text{d}s}\tilde{\phi}_i(s)
    \end{align}
    To evaluate $\partial_\lambda\tilde{G}_i^{(\lambda)}(s)$, we use the expression from Proposition \ref{Prop:time_additive_observable_mean}:
    \begin{align}
        \partial_\lambda\tilde{G}_i^{(\lambda)}(s)\big|_{\lambda=0}&=\frac{1}{s}\sum_{l\in V(\graph)}g_l\frac{\sum_{\mathcal{F}^{s,i\to l}_{[l,\ast]}}w(\mathcal{F})}{\sum_{\mathcal{T}^{s}_{[\ast]}}w(\mathcal{T})},
    \end{align}
    so:
    \begin{align}
        2\big(\partial_\lambda\tilde{Q}^{(\lambda)}(s)\partial_\lambda\tilde{G}^{(\lambda)}(s)\big)_{i}\big|_{\lambda=0}&=-\frac{2g_i}{s}\sum_{j,l\in V(\graph)}g_l\big(\frac{\text{d}}{\text{d}s}\Tilde{Q}_{ij}(s)\big)\frac{\sum_{\mathcal{F}^{s,j\to l}_{[l,\ast]}}w(\mathcal{F})}{\sum_{\mathcal{T}^{s}_{[\ast]}}w(\mathcal{T})}.
    \end{align}
    Combining these terms and expressing the inverse of $s\mathbb{1}-\tilde{W}^T(s)$ as in Proposition \ref{Prop:inverse_of_s-W}, we obtain:
    \begin{align}
        \frac{\partial^2}{\partial\lambda^2}\tilde{G}_i(\lambda,s)\Big|_{\lambda=0}=&\frac{2}{s}\sum_{j\in V(\graph)}g_j\frac{\sum_{\mathcal{F}^{s,i\to j}_{[j,\ast]}}w(\mathcal{F})}{\sum_{\mathcal{T}^{s}_{[\ast]}}w(\mathcal{T})}\Big(\frac{g_j}{s}\big(1+\frac{\frac{\text{d}}{\text{d}s}\Tilde{Q}_j(s)}{\tilde{\phi}_j(s)}\big)\\&-\sum_{m,l\in V(\graph)}g_l\big(\frac{\frac{\text{d}}{\text{d}s}\Tilde{Q}_{jm}(s)}{\tilde{\phi}_j(s)}\big)\frac{\sum_{\mathcal{F}^{s,m\to l}_{[l,\ast]}}w(\mathcal{F})}{\sum_{\mathcal{T}^{s}_{[\ast]}}w(\mathcal{T})}\Big).
    \end{align}
    Eq.~(\ref{eq:transient_semi_MCTT}) implies that the ratios of the sums over the spanning forests and spanning trees can be interpreted as the Laplace transforms of the occupation probabilities under suitable initial conditions. For the $s\to 0$ limit, this implies that:
    \begin{align}
        \frac{\sum_{\mathcal{F}^{s,i\to j}_{[j,\ast]}}w(\mathcal{F})}{\sum_{\mathcal{T}^{s}_{[\ast]}}w(\mathcal{T})}=\frac{\pi_j}{s}+o(s^{-1}).
    \end{align}

    Therefore, we get the following small-$s$ regime:
    \begin{align}
        \frac{\partial^2}{\partial\lambda^2}\tilde{G}_i(\lambda,s)\Big|_{\lambda=0}=&\frac{2}{s^3}\sum_{j\in V(\graph)}g_j\pi_j\Big(g_j\big(1-\alpha_j\big)+\sum_{l\in V(\graph)}g_l\pi_l\alpha_j\Big)+o(s^{-3}),
    \end{align}
    using that $\lim_{s\to 0}\frac{\text{d}}{\text{d}\ln s}\ln\tilde{K}_{ij}(s)=1-\alpha_i=1+\lim_{s\to 0}\frac{\frac{\text{d}}{\text{d}s}\Tilde{Q}_i(s)}{\tilde{\phi}_i(s)}$, by assumption and by Lemma \ref{lemma:memory_kernels_small_s}.

    Now define $g_b:=-\min_{i\in V(\mathcal{G})}g_i$ and consider  
    \begin{align}
        B(t):=\integ{0}{t}{(g(X_{t'})+g_b)}{t'}=A(t)+g_bt
    \end{align}
    Note that 
    \begin{align}
        \frac{B(t)}{t}=\frac{A(t)}{t}+g_b\Rightarrow \text{Var}\big(\frac{B(t)}{t}\big)=\text{Var}\big(\frac{A(t)}{t}\big).
    \end{align}
    Moreover, $B(t)^2$, and hence $\mathbb{E}[B(t)^2]$, is non-decreasing in $t$ since $g(X_{t})+g_b\geq 0$, so Karamata’s Tauberian theorem can be applied, yielding the following long-time scaling for the second moment of $B(t)$:
    \begin{align}
        \mathbb{E}[B(t)^2]=\sum_{j\in V(\graph)}\pi_j(g_j+g_b)\Big((g_j+g_b)\big(1-\alpha_j\big)+\sum_{l\in V(\graph)}(g_l+g_b)\pi_l\alpha_j\Big)t^2+o(t^2).
    \end{align}
    By Proposition \ref{prop:time_average_mean_converges}, we have:
    \begin{align}
        \mathbb{E}[B(t)]^2&=\Big(\sum_{j\in V(\graph)}\pi_j(g_j+g_b)\Big)^2t^2+o(t^2)\\ &=\big(\sum_{j,l\in V(\graph)}\pi_j\pi_l(g_j+g_b)(g_l+g_b)\big)t^2+o(t^2).
    \end{align}
    Thus, we obtain:
    \begin{align}
        \lim_{t\to \infty}\text{Var}\big(\frac{A(t)}{t}\big)&= \lim_{t\to \infty}\text{Var}\big(\frac{B(t)}{t}\big)\\&=\lim_{t\to \infty}\frac{\mathbb{E}[B(t)^2]-\mathbb{E}[B(t)]^2}{t^2}\\&=\sum_{j\in V(\graph)}\pi_j(g_j+g_b)\Big((g_j+g_b)\big(1-\alpha_j\big)-\sum_{l\in V(\graph)}(g_l+g_b)\pi_l(1-\alpha_j)\Big)\\ &=\sum_{j\in V(\graph)}\pi_j(g_j+g_b)\big(1-\alpha_j\big)\sum_{l\in V(\graph)}(g_j-g_l)\pi_l.
    \end{align}
   By Lemma \ref{lemma:memory_kernels_small_s}, there is an $\alpha\in[0,1]$ such that $\pi_j>0$ implies $\alpha_j=\alpha$, which yields:
   \begin{align}
        \lim_{t\to \infty}\text{Var}\big(\frac{A(t)}{t}\big)&=\big(1-\alpha\big)\Big(\sum_{j\in V(\graph)}\pi_jg_j^2-\big(\sum_{j\in V(\graph)}\pi_jg_j\big)^2\Big),
    \end{align}
    which is zero if and only if $\alpha=1$, as in the finite-mean sojourn time case, or if the variance of $g_j$ with respect to the limiting state distribution $\pi$ vanishes, i.e., $g_i=g_j$ wherever $\pi_i,\pi_j>0$. 

     Since positive limiting variance is enough to rule out convergence in probability, it also rules out almost-sure convergence to a constant, and since the mean of the time average always converges according to Proposition \ref{prop:time_average_mean_converges}, this implies that $A(t)$ exhibits weak ergodicity breaking.
\end{proof}

Theorem \ref{theorem:weak_ergodicity_breaking} shows that the variance of $\frac{A(t)}{t}$ vanishes for $t\to\infty$ in the case of finite-mean sojourn times ($\alpha=1$) but does not make a statement about almost-sure convergence. Ergodicity is a standard result in the finite-mean sojourn time case and can, e.g., be found in \cite{athreya1978limit1,athreya1978limit}.

\new{The expression for the variance in Eq.~(\ref{eq:convergence_variance}) was previously derived by Rebenshtok and Barkai~\cite{rebenshtok2007distribution,rebenshtok2008weakly} 
for subdiffusive continuous-time random walks with a common heavy-tail 
exponent, as part of their more general characterization of the limiting distributions of time-averaged observables. 
\newnew{They also discussed the extension of their results to 
state-dependent waiting-time distributions with different heavy-tail 
exponents~\cite{rebenshtok2008weakly}.}
\newnew{Here, we obtained the same relation for finite semi-Markov 
networks with potentially heterogeneous and transition-dependent 
waiting-time distributions within the graph-theoretic framework}, 
under the assumptions of Theorem~\ref{theorem:weak_ergodicity_breaking}. 
In particular, the derivation only requires the 
\newnew{existence of the small-$s$ limits of the logarithmic derivatives of the memory kernels rather than an explicit power-law expansion,} and the existence of a limit of $s\tilde p(s)$, which is 
independent of the initial condition, without assuming pointwise 
convergence of the occupation probabilities or deriving the full 
limiting distribution of the time average.}

\section{Example: two-state semi-Markov network}
Finally, we shall illustrate the results of this work for the simplest non-trivial example: a semi-Markov network on two states with back-and-forth transitions, i.e.,
\begin{align}
    \graph=(V(\graph),E(\graph))=\big(\{1,2\},\{(1,2),(2,1)\}\big).
\end{align}
Then $p_{12}=p_{21}=1$, $Q_{12}=Q_1$ and $Q_{21}=Q_2$, so the memory kernels read for \textit{a priori} arbitrary sojourn time distributions:
\begin{align}
    \tilde{K}_{12}(s)=\frac{s\tilde{Q}_{1}(s)}{1-\tilde{Q}_{1}(s)},\qquad \tilde{K}_{21}(s)=\frac{s\tilde{Q}_{2}(s)}{1-\tilde{Q}_{2}(s)}.
\end{align}
Eq.~(\ref{eq:generalized_ME_Laplace}) then yields:
\begin{align}
    \begin{pmatrix}p_{0,1}\\p_{0,2}\end{pmatrix}=(s\mathbb{1}-\tilde{W}(s))\begin{pmatrix}\tilde{p}_{1}(s)\\\tilde{p}_{2}(s)\end{pmatrix}=\begin{pmatrix}s+\tilde{K}_{12}(s) &-\tilde{K}_{21}(s)\\-\tilde{K}_{12}(s)& s+\tilde{K}_{21}(s)\end{pmatrix}\begin{pmatrix}\tilde{p}_{1}(s)\\\tilde{p}_{2}(s)\end{pmatrix}\\ \Rightarrow \begin{pmatrix}\tilde{p}_{1}(s)\\\tilde{p}_{2}(s)\end{pmatrix}=\frac{1}{s\big(s+\tilde{K}_{12}(s)+\tilde{K}_{21}(s)\big)}\begin{pmatrix}s+\tilde{K}_{21}(s) &\tilde{K}_{21}(s)\\\tilde{K}_{12}(s)& s+\tilde{K}_{12}(s)\end{pmatrix}\begin{pmatrix}p_{0,1}\\p_{0,2}\end{pmatrix}.\label{eq:two_state_occupation_probabilities}
\end{align}
The augmented graph of this network is shown in Fig.~\ref{fig:two_state_augmented_graph}(a) and its spanning trees and spanning forests arising in the combinatorial interpretation of the resolvent are depicted in Fig.~\ref{fig:two_state_augmented_graph}(b). Summing over their weights yields:
\begin{align}
    \sum_{\mathcal{T}^{s}_{[\ast]}}w(\mathcal{T})&=s^2+s\tilde{K}_{12}(s)+s\tilde{K}_{21}(s)=s(s+\tilde{K}_{12}(s)+\tilde{K}_{21}(s)),\\ &\sum_{\mathcal{F}^{s,1\to 1}_{[1,\ast]}}w(\mathcal{F})=s+\tilde{K}_{21}(s),~~\sum_{\mathcal{F}^{s,2\to 1}_{[1,\ast]}}w(\mathcal{F})=\tilde{K}_{21}(s)\\ &\sum_{\mathcal{F}^{s,1\to 2}_{[2,\ast]}}w(\mathcal{F})=\tilde{K}_{12}(s),~~\sum_{\mathcal{F}^{s,2\to 2}_{[2,\ast]}}w(\mathcal{F})=s+\tilde{K}_{12}(s)  
\end{align}
verifying Eq.~(\ref{eq:transient_semi_MCTT}).

\begin{figure}[t]
    \centering
    \includegraphics[width=0.9\linewidth]{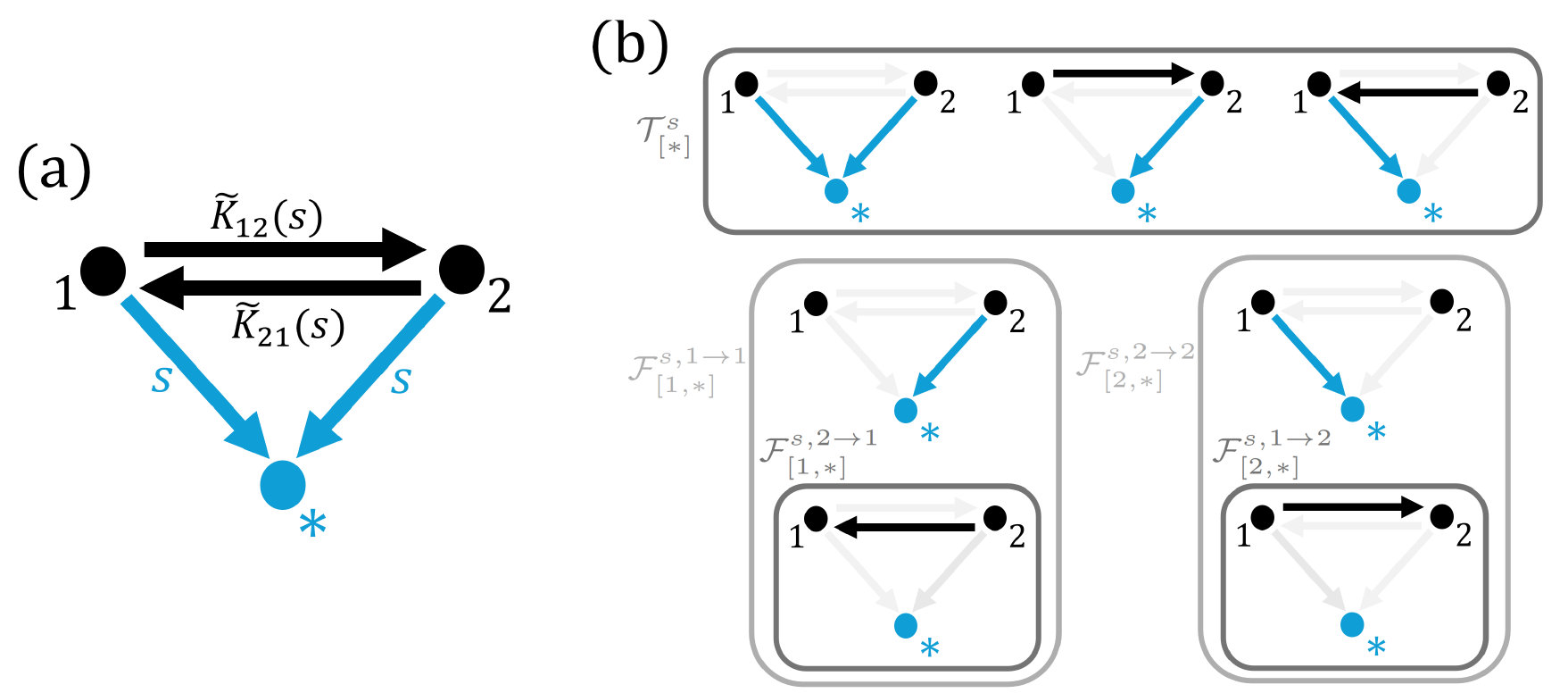}
    \caption{Illustration of the augmented graph of a semi-Markov network with two states. (a) The network on the augmented graph $\graph^\ast=(V(\graph^\ast),E(\graph^\ast))=(\{1,2,\ast\},\{(1,2),(2,1),(1,\ast),(2,\ast))\})$. The edges leading to $\ast$ are given a weight $s$. (b) The spanning trees of $\graph^\ast$ rooted at $\ast$ (upper panel) and the spanning forests with roots $1,\ast$ and $2,\ast$ (lower panel) which correspond to the combinatorial terms to explicitly express the resolvent.}
    \label{fig:two_state_augmented_graph}
\end{figure}

For finite-mean sojourn times, we have
\begin{align}
    \tilde{K}_{12}(s)=\mathbb{E}[\tau_1]^{-1}+o(1),~~\tilde{K}_{21}(s)=\mathbb{E}[\tau_2]^{-1}+o(1),
\end{align}
so, $\kappa_1=\mathbb{E}[\tau_1]^{-1}$, $\kappa_2=\mathbb{E}[\tau_2]^{-1}$ and the stationary limit of Eq.~(\ref{eq:two_state_occupation_probabilities}) becomes:
\begin{align}
    \begin{pmatrix}\pi_{1}\\\pi_{2}\end{pmatrix}&=\lim_{s\to 0}s\begin{pmatrix}\tilde{p}_{1}(s)\\\tilde{p}_{2}(s)\end{pmatrix}\\&=\frac{1}{\kappa_1+\kappa_2}\begin{pmatrix}\kappa_2 &\kappa_2\\\kappa_1& \kappa_1\end{pmatrix}\begin{pmatrix}p_{0,1}\\p_{0,2}\end{pmatrix}\\&=\frac{1}{\kappa_1+\kappa_2}\begin{pmatrix}\kappa_2\\\kappa_1\end{pmatrix}\\&=\frac{1}{\mathbb{E}[\tau_{1}]+\mathbb{E}[\tau_{2}]}\begin{pmatrix}\mathbb{E}[\tau_{1}] \\\mathbb{E}[\tau_{2}] \end{pmatrix},
\end{align}
using that $p_{0,1}+p_{0,2}=1$, in agreement with Eq.~(\ref{eq:stationary_solution}), as there are only two spanning trees, which yield:
\begin{align}
    \sum_{\mathcal{T}_{[1]}}w(\mathcal{T})=\kappa_2,~~\sum_{\mathcal{T}_{[2]}}w(\mathcal{T})=\kappa_1.
\end{align}

Similarly, for infinite-mean sojourn times with the small-$s$ scalings
\begin{align}
    \tilde{K}_{12}(s)=\kappa_1s^{1-\alpha_1}+o(s^{1-\alpha_1}),~~\tilde{K}_{21}(s)=\kappa_2s^{1-\alpha_2}+o(s^{1-\alpha_2}),
\end{align}
we have for $\alpha_1<\alpha_2$:
\begin{align}
    \begin{pmatrix}\pi_{1}\\\pi_{2}\end{pmatrix}&=\lim_{s\to 0}s\begin{pmatrix}\tilde{p}_{1}(s)\\\tilde{p}_{2}(s)\end{pmatrix}=\begin{pmatrix}1\\0\end{pmatrix},
\end{align} 
and for $\alpha_1=\alpha_2=\alpha$:
\begin{align}
    \begin{pmatrix}\pi_{1}\\\pi_{2}\end{pmatrix}&=\lim_{s\to 0}s\begin{pmatrix}\tilde{p}_{1}(s)\\\tilde{p}_{2}(s)\end{pmatrix}=\frac{1}{\kappa_1+\kappa_2}\begin{pmatrix}\kappa_2\\\kappa_1\end{pmatrix}.
\end{align}

Convergence to these limiting occupation distributions is illustrated in Fig.~\ref{fig:convergence}. For the finite-mean case, gamma-distributed sojourn times are chosen, i.e., they have probability densities of the shape:
\begin{align}
    \psi_i(t)=\frac{t^{r_i-1}}{\Gamma(r_i)\theta_i^{r_i}}e^{-\frac{t}{\theta_i}},
\end{align}
with mean sojourn time
\begin{align}
    \mathbb{E}[\tau_i]=r_i\theta_i,
\end{align}
for $i\in\{1,2\}$. For the infinite-mean case, the survival probabilities are chosen to be:
\begin{align}
    \phi_i(t)=(1+\frac{t}{\tau_i})^{-\alpha_i},
\end{align}
for $\alpha_i\in(0,1)$, corresponding to a Pareto type II distribution.
\begin{figure}
    \centering
    \includegraphics[width=0.8\linewidth]{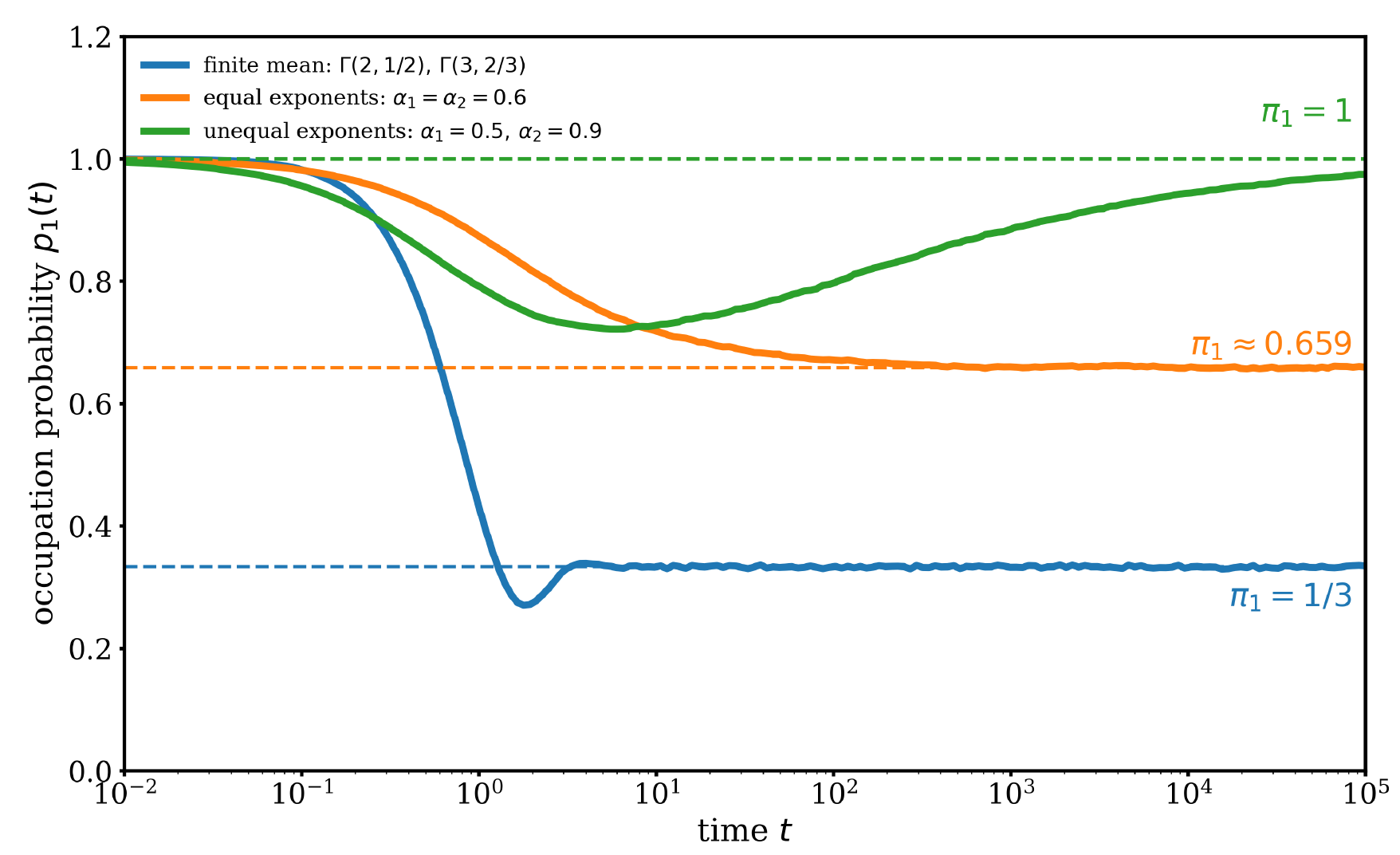}
    \caption{Illustration of the convergence to the limiting occupation distribution for two-state semi-Markov networks. The plot shows Monte Carlo estimates of the occupation probability of state $1$ $p_1(t)$ from $10^5$ trajectories and the dashed lines show the analytical solutions for the limiting occupation distribution. For the finite-mean case (blue curve), gamma distributions were chosen with $(r_1,\theta_1)=(2,\frac{1}{2})$ and $(r_2,\theta_2)=(3,\frac{2}{3})$ for $Q_1$ and $Q_2$, respectively. The mean sojourn times are therefore $\mathbb{E}[\tau_1]=1$ and $\mathbb{E}[\tau_2]=2$ and the stationary occupation probability reads $\pi_1=\frac{\mathbb{E}[\tau_1]}{\mathbb{E}[\tau_1]+\mathbb{E}[\tau_2]}=\frac{1}{3}$. For the heavy-tailed cases, the survival probabilities are $\phi_i(t)=(1+\frac{t}{\tau_i})^{-\alpha_i}$ for states $i\in\{1,2\}$. In the case of equal exponents with $\alpha=\alpha_1=\alpha_2=0.6$ and $(\tau_1,\tau_2)=(3,1)$ (orange curve), the limiting probability is $\pi_1=\frac{\tau_1^\alpha}{\tau_1^\alpha+\tau_2^\alpha}\approx 0.659$, while for unequal exponents, $\alpha_1=0.5$ and $\alpha_2=0.9$ and $(\tau_1,\tau_2)=(1,1)$ (green curve), the limiting probability is $\pi_1=1$.}
    \label{fig:convergence}
\end{figure}
The dashed lines indicate the exact asymptotic values from Theorem \ref{theorem:SMCTT}, while the solid curves are Monte Carlo estimates of the occupation probability $p_1(t)$ based on $10^5$ trajectories. For the case of infinite-mean sojourn times, we have $\pi_1,\pi_2>0$ for equal exponents for $\alpha_1=\alpha_2$, and $\pi_1=1,\pi_2=0$ for $\alpha_1<\alpha_2$.

We can also illustrate weak ergodicity breaking in the case of infinite-mean sojourn times. To this end, consider the empirical occupation fraction of state 1,
\begin{align}
    \rho_1(t):=\frac{1}{t}\integ{0}{t}{\mathbb{1}_{\{X_{t'}~=~1\}}}{t'},
\end{align}
which corresponds to an additive state observable $A(t)$ with $g_1=1,g_2=0$ divided by $t$.  
In the two-state case, $\rho(t)$ converges for $t\to\infty$ in distribution to a generalized arcsine distribution with density given by \cite{bel2005weak}
\begin{align}
    f_{\alpha,R}(p)=\frac{\sin{(\pi \alpha)}}{\pi}\frac{Rp^{\alpha-1}(1-p)^{\alpha-1}}{R^2(1-p)^{2\alpha}+p^{2\alpha}+2Rp^\alpha(1-p)^{\alpha}\cos{(\pi \alpha)}},
\end{align}
where $R=\frac{\kappa_2}{\kappa_1}$. Its first moment and variance read~\cite{akimoto2023statistics}:
\begin{align}
    \mathbb{E}[\rho]=\frac{\kappa_2}{\kappa_1+\kappa_2}=\pi_1,~~\text{Var}(\rho)=(1-\alpha)\frac{\kappa_1\kappa_2}{(\kappa_1+\kappa_2)^2}=(1-\alpha)\pi_1\pi_2,
\end{align}
in agreement with Eqs.~(\ref{eq:convergence_mean}) and (\ref{eq:convergence_variance}), respectively.

For the hitting-time distribution for going from $1$ to $2$ with $p_{0,1}=1$, we have $\tilde{f}(s)=\tilde{Q}_1(s)$ since the hitting time in this case simply coincides with the sojourn time in the initial state. Applying Eq.~(\ref{eq:FPT_distribution_graph_theory_1}), we also obtain:
\begin{align}
    \tilde{f}(s)=\frac{\Tilde{K}_{12}(s)\sum_{\mathcal{F}^{s,1\to 1}_{[1,2,\ast]}}w(\mathcal{F})}{\sum_{\mathcal{F}^{s}_{[2,\ast]}}w(\mathcal{F})}=\frac{\Tilde{K}_{12}(s)}{s+\Tilde{K}_{12}(s)}=\tilde{Q}_1(s),
\end{align}
using that $w(\{1,2\},\emptyset)=1$ and $\tilde{K}_{12}(s)=\frac{s\tilde{Q}_1(s)}{1-\tilde{Q}_1(s)}$.

\section{Conclusion}\label{sec:conclusion}

In this work, the graph-theoretic framework for Markov networks has been extended to finite semi-Markov networks, allowing us to express quantities such as occupation probabilities or hitting-time distributions in Laplace space and their moments as combinatorial sums over the weights of spanning trees and forests of an augmented graph. 
These explicit expressions are particularly useful for showing asymptotic results, such as the shape of the stationary occupation probabilities, the long-time scaling of the hitting-time distribution or weak ergodicity breaking. Moreover, these expressions may facilitate proofs or extend their scope to semi-Markov networks, as illustrated in the case of Kemeny's constant.

\new{The special case of Markov networks follows for constant weights, which was previously derived to provide a graph-theoretic proof for mutual linearity in the transient regime~\cite{voits2026graph}. This relation also provided explicit, graph-theoretic interpretations of the involved coefficients such as the susceptibility. While mutual linearity does not extend to all local perturbations in semi-Markov networks, the graph-theoretic framework also provides means to give explicit expressions for the corresponding coefficients for a broad class of perturbations~\cite{voits2026mutuallinearitysemimarkovnetworks}\newnew{Moreover, the Laplace transforms of the occupation probabilities and 
FPT distributions acquire a direct physical meaning in the context of stochastic resetting: for a process subject to Poissonian resetting at rate $s$, $s\tilde{p}_i(s)$ and $\tilde{f}_i(s)$ correspond to the occupation probability of state $i$ immediately before resetting and the probability of reaching the target state within a single attempt, respectively~\cite{evans2020stochastic}. For stochastic resetting on semi-Markov networks, the graph-theoretic framework provides explicit expressions for these quantities.}

However, a limitation of this method is the rapidly growing number of spanning trees and forests with system size, making their explicit enumeration inefficient for solving concrete models beyond the simplest topologies. Therefore, the framework's main strength resides in its use as an analytical tool to derive exact relations rather than a direct method for numerical evaluation of the involved quantities.}

In Markov networks, the spectral properties of the generator and, in particular, the spectral gap play an important role~\cite{norris1998markov,levin2026markov}. Since the characteristic polynomial coincides with the determinant of the inverse resolvent, 
\begin{align}
    \chi_W(s)=\det{(s\mathbb{1}-W)}=\sum_{\mathcal{T}^s_{[\ast]}}w(\mathcal{T}),
\end{align}
where 
\begin{align}
    w(\mathcal{T})=\big(\prod_{\substack{(i,j)\in E(\mathcal{T}),~j\neq \ast}}k_{ij}\big)s^{|\{i\in V(\mathcal{T})~|~(i,\ast)\in E(\mathcal{T})\}|},
\end{align}
a natural generalization to semi-Markov networks would be to consider instead
\begin{align}
    \det{\big(s\mathbb{1}-\tilde{W}(s)\big)}=\sum_{\mathcal{T}^s_{[\ast]}}w(\mathcal{T}),
\end{align}
with
\begin{align}
    w(\mathcal{T})=\big(\prod_{\substack{(i,j)\in E(\mathcal{T}),~j\neq \ast}}\tilde{K}_{ij}(s)\big)s^{|\{i\in V(\mathcal{T})~|~(i,\ast)\in E(\mathcal{T})\}|}.
\end{align}
It should be noted that this expression was introduced here for $s\in (0,\infty)$ and studying its zeros as a generalization of the roots of the characteristic polynomial would require extending the expression to the complex left half-plane. However, in the case of sojourn times with infinite moments, branch points arise due to the non-integer scaling of $\tilde{K}_{ij}(s)$ for $s\to 0$, which makes this extension more delicate and requires a more thorough treatment in the future.

In conclusion, the graph-theoretic framework appears to be a promising approach to the study of semi-Markov networks, which may contribute to a better understanding of the structural similarities of Markov and semi-Markov networks and how they depend on the sojourn-time distributions. Beyond extending mathematical results, it may also be useful for analyzing coarse-grained stochastic dynamics as they arise, e.g., in biomolecular and cellular state transitions, transport through disordered media or credit-rating dynamics.

\backmatter

\bmhead{Acknowledgements}

I would like to thank Prof. Ulrich Schwarz for support and for helpful comments. In addition, I gratefully acknowledge the support of the German Academic Scholarship Foundation (Studienstiftung des deutschen Volkes).

\begin{appendices}
\setcounter{theorem}{0}
\renewcommand{\thetheorem}{A.\arabic{theorem}}
\setcounter{equation}{0}
\renewcommand{\theequation}{A.\arabic{equation}}

\section{Laplace–Stieltjes transforms and asymptotic limits}\label{appendix_Laplace_Stieltjes}
In this appendix, we collect the relevant theory of Laplace–Stieltjes transforms and asymptotic results used throughout the paper.
\begin{definition}
    Let $\mathcal{B}([0,\infty))$ be the Borel-$\sigma$-algebra on $[0,\infty)$ and let $Q:\mathcal{B}([0,\infty))\to [0,\infty)$ be a finite measure. The \textit{Laplace–Stieltjes transform} of $Q$ is defined as
    \begin{align}
        \tilde{Q}(s)=\integ{[0,\infty)}{}{e^{-st}Q(}{t)},
    \end{align}
    for $s\in (0,\infty)$.
\end{definition}

If $Q$ admits a density, $Q(\text{d}t)=\psi(t)\text{d}t$,
this reduces to the standard Laplace transform of $\psi$:
\begin{align}
    \tilde{Q}(s)=\integ{[0,\infty)}{}{e^{-st}Q(}{t)}=\integ{0}{\infty}{\psi(t)e^{-st}}{t}.
\end{align}

The moments of the distribution can be calculated by evaluating the derivatives of the Laplace–Stieltjes transform in $s=0$:

\begin{lemma}
    Let $Q$ be a probability measure on $[0,\infty)$ and let $\tilde{Q}(s)$ be its Laplace–Stieltjes transform. Then the moments of $\tau\sim Q$ are given by
    \begin{align}
        \mathbb{E}[\tau^n]=(-1)^n\lim_{s\to0}\frac{\text{d}^n}{\text{d}s^n}\tilde{Q}(s),
    \end{align}
    for $n\in \mathbb{N}_0$.
\end{lemma}
\begin{proof}
    Observe that for $n\in \mathbb{N}_0$,
    \begin{align}
       \frac{\text{d}^n}{\text{d}s^n}\tilde{Q}(s)=(-1)^n\integ{[0,\infty)}{}{t^ne^{-st}Q(}{t)}.    
    \end{align}
    Hence, by monotone convergence, we have
    \begin{align}
       \lim_{s\to0}\frac{\text{d}^n}{\text{d}s^n}\tilde{Q}(s)=(-1)^n\integ{[0,\infty)}{}{t^nQ(}{t)}=\mathbb{E}[\tau^n].    
    \end{align}
\end{proof}

We used the small $s$-asymptotics in Laplace space to obtain the long-time limits 
in time domain. This was based on the following standard results:

\begin{lemma}
    Let $f:[0,\infty)\to [0,\infty)$ be a locally integrable function for $t\geq 0$ and $\tilde{f}(s)$ its Laplace transform for $s>0$. If $\lim_{t\to \infty}f(t)$ exists and is finite, then
    \begin{align}
        \lim_{s\to 0}s\tilde{f}(s)=\lim_{t\to \infty}f(t).
    \end{align}
\end{lemma}
\begin{proof}
    Let $c:=\lim_{t\to \infty}f(t)<\infty$. Then, for every $\epsilon >0$, there exists $t_\epsilon\in (0,\infty)$ such that
    \begin{align}
        \big|f(t)-c\big|<\epsilon
    \end{align}
    for $t>t_\epsilon$. Consider,
    \begin{align}
        s\tilde{f}(s)=s\integ{0}{t_\epsilon}{f(t)e^{-st}}{t}+s\integ{t_\epsilon}{\infty}{f(t)e^{-st}}{t}.
    \end{align}
    Since $\integ{0}{t_\epsilon}{f(t)}{t}<\infty$, 
    \begin{align}
        \lim_{s\to 0}s\integ{0}{t_\epsilon}{f(t)e^{-st}}{t}=0.
    \end{align}
    For the second integral, the following bounds must hold:
    \begin{align}
        (c-\epsilon)e^{-st_\epsilon} \leq s\integ{t_\epsilon}{\infty}{f(t)e^{-st}}{t}\leq (c+\epsilon)e^{-st_\epsilon},
    \end{align}
    so in the limit $s\to 0$, we get
    \begin{align}
        c-\epsilon \leq \lim_{s\to 0}s\integ{t_\epsilon}{\infty}{f(t)e^{-st}}{t}\leq c+\epsilon.
    \end{align}
    Since $\epsilon>0$ was arbitrary, this implies
    \begin{align}
        \lim_{s\to 0}s\tilde{f}(s)=\lim_{s\to 0}s\integ{t_\epsilon}{\infty}{f(t)e^{-st}}{t}=c=\lim_{t\to \infty}f(t).
    \end{align}
\end{proof}

The lemma assumes the existence of $\lim_{t\to \infty}f(t)$ to deduce the limit $\lim_{s\to 0}s\tilde{f}(s)$, which is known as the Abelian direction of the final value theorem. The converse is not always true, i.e., the existence of $\lim_{s\to 0}s\tilde{f}(s)$ need not imply the existence of $\lim_{t\to \infty}f(t)$, in general, but additional conditions are required to assure it. Deducing the long-time asymptotics of $f(t)$ from the small-$s$ behavior of $\tilde{f}(s)$ is known as the Tauberian direction. There are different results known for this direction and the one we used is known as Karamata's Tauberian theorem:

\begin{theorem}[Karamata]\label{Theorem:Karamata}
    Let $f:[0,\infty)\to [0,\infty)$ be a locally integrable function and ultimately monotone. Let $\tilde{f}(s)$ be its Laplace transform. Then the following holds:
    \begin{align}
        f(t)\sim c\frac{t^{\alpha-1}}{\Gamma(\alpha)}~~\text{for }t\to \infty \Leftrightarrow \tilde{f}(s)\sim cs^{-\alpha}~~\text{for }s\to 0,
    \end{align}
    for $c>0$ and $\alpha >0$.
\end{theorem}
\begin{proof}
    See, e.g., \cite{feller1991introduction}.
\end{proof}

The monotonicity assumption is automatically satisfied when
$f$ is a survival probability.

\new{\section{Proofs for Section \ref{sec:basics}}
\label{appendix_proofs}

\begin{proof}{\textit{of Lemma \ref{lemma:Q_W_identity}:}}
    Using that $\forall i\in V^\circ(\graph)$:
    \begin{align}
        \tilde{\phi}_i(s)=&\frac{1}{s}\big(1-\tilde{Q}_i(s)\big),
    \end{align}
    and $\tilde{\phi}_i(s)=\frac{1}{s}$ for $i\in V(\graph)\setminus V^\circ(\graph)$, the identity can be explicitly verified: 
    \begin{align}
        \Big[(s\mathbb{1}-\tilde{W}(s))\diag{\tilde{\phi}(s)}\Big]_{ij}&=-\tilde{K}_{ji}(s)\tilde{\phi}_j(s)\\ &=-\tilde{Q}_{ji}(s),
    \end{align}
    and for $i=j\in V^\circ(\graph)$:
        \begin{align}
        \Big[(s\mathbb{1}-\tilde{W}(s))\diag{\tilde{\phi}(s)}]_{ii}&=\Big(s+\sum_{l=1}^N\tilde{K}_{il}(s)\Big)\frac{1-\tilde{Q}_i(s)}{s}\\ &=1-\tilde{Q}_i(s)+\sum_{l=1}^N\tilde{Q}_{il}(s)\\ &=1.
    \end{align}
    For absorbing states, $\tilde{Q}_i(s)=0$, $\tilde {K}_{il}(s)=0$, and $\tilde\phi_i(s)=\frac{1}{s}$ yield the same identity.

    Lastly, since $\forall i\in V(\graph)$ $\phi_i\geq0 $ and $\phi_i(0) =1 $, each Laplace transform is strictly positive, $\diag{\tilde{\phi}(s)}$ is invertible.
\end{proof} 

\begin{proof}{\textit{of Proposition \ref{prop:generalized_ME_Laplace}:}}
    Consider
    \begin{align}
        R_i(B):=\mathbb{E}\big[\sum_{n\geq 0}\mathbb{1}_{\{X_n^d=i,\ t_n\in B\}}\big]
    \end{align}
    for Borel sets $B\subseteq[0,\infty)$. Then
    \begin{align}
         R_i(B)=p_{0,i}\delta_0(B)+\sum_{j\in V(\graph)}(R_j*Q_{ji})(B).
    \end{align}
    In Laplace space, this yields
    \begin{align}
        \tilde{R}_i(s)=p_{0,i}+\sum_{j\in V(\graph)}\tilde{Q}_{ji}(s)\tilde{R}_j(s).
    \end{align}
    Writing $(\tilde{R}(s))_i=\tilde{R}_i(s)$, $\tilde{R}(s)$ can therefore be expressed as:
    \begin{align}
        (\mathbb{1}-\tilde{Q}^T(s))\tilde{R}(s)=p_0. \label{eq:R_Laplace_space}      
    \end{align}
    The occupation probabilities then follow as:
    \begin{align}
         p_i(t)=\int_{[0,t]}\phi_i(t-u)R_i(du),
    \end{align}
    which in Laplace space reads:
    \begin{align}
        \tilde{p}(s)=\diag{\tilde{\phi}(s)}\tilde{R}(s),
    \end{align}
    where $(\tilde{p}(s))_i=\tilde{p}_i(s)$ and $\diag{\tilde{\phi}(s)}:=\diag{\tilde{\phi}_1(s),...,\tilde{\phi}_N(s)}$. Combining this expression with Eq.~(\ref{eq:R_Laplace_space}) yields
    \begin{align}
        (\mathbb{1}-\tilde{Q}^T(s))\diag{\tilde{\phi}_1(s),...,\tilde{\phi}_N(s)}^{-1}\tilde{p}(s)=p_0,
    \end{align}
    using that $\tilde{\phi}_i(s)>0$ since $\phi_i\geq 0$ and $\phi_i(0)=1$. 
    Applying Lemma \ref{lemma:Q_W_identity}, we can alternatively express this equation as:
    \begin{align}
        (s\mathbb{1}-\tilde{W}(s))\tilde{p}(s)=p_0.
    \end{align}
\end{proof}

\section{Proof of Lemma \ref{lemma:finite_K}}\label{appendix_proof_Lemma_14}
\begin{proof}{\textit{of Lemma \ref{lemma:finite_K}:}}
    A finite $n$-th moment implies that also all moments of order $m\leq n$ are finite, and since $Q_{ij}\leq Q_i$, we also have 
    \begin{align}
        \integ{0}{\infty}{t^mQ_{ij}(}{t})<\infty,
    \end{align}
    for $m\leq n$.

    For the Laplace transforms, this means that $\forall m\in\{1,...,n\}$:
    \begin{align}
        \lim_{s\to 0}\frac{\text{d}^m\Tilde{Q}_{ij}}{\text{d}s^m}(s)&=(-1)^m\integ{0}{\infty}{t^mQ_{ij}(}{t})<\infty\\ \lim_{s\to 0}\frac{\text{d}^m\Tilde{Q}_{i}}{\text{d}s^m}(s)&=(-1)^m\integ{0}{\infty}{t^mQ_{i}(}{t})<\infty.
    \end{align}
    Since $Q_i\geq 0$, the following holds for $m\leq n-1$:
    \begin{align}
      \lim_{s\to 0}\frac{\text{d}^{m}\tilde{\phi}_i}{\text{d}s^m}(s)&=(-1)^m\integ{0}{\infty}{t^m\phi_i(t)}{t}\\&=(-1)^m\integ{0}{\infty}{t^m\integ{0}{\infty}{\mathbb{1}_{\{t'>t\}}Q_i(}{t'})}{t}\\&=(-1)^m\integ{0}{\infty}{\Big(\integ{0}{t'}{t^m}{t}\Big)Q_i(}{t')}\\&=\frac{(-1)^m}{m+1}\integ{0}{\infty}{t'^{m+1}Q_i(}{t')}<\infty\label{eq:relation_phi_psi_moments}.
    \end{align}
    Therefore, since $\tilde{\phi}_i(0)>0$ and $\tilde{K}_{ij}(s)=\frac{\Tilde{Q}_{ij}(s)}{\tilde{\phi}_i(s)}$, for all $m\leq n-1$, $\lim_{s\to 0}\frac{\text{d}^m}{\text{d}s^m}\tilde{K}_{ij}(s)$ exists and is finite.

    For the final claim, it remains to show that $\lim_{s\to 0}s\frac{\text{d}^n\tilde{\phi}_i}{\text{d}s^n}(s)=0$. Note that
    \begin{align}
        s\frac{\text{d}^n\tilde{\phi}_i}{\text{d}s^n}(s)&=-\big(n\frac{\text{d}^{n-1}\tilde{\phi}_i}{\text{d}s^{n-1}}(s)+\frac{\text{d}^{n}\Tilde{Q}_i}{\text{d}s^{n}}(s)\big),
    \end{align}
    which follows by induction as
    \begin{align}
        s\frac{\text{d}\tilde{\phi}_i}{\text{d}s}(s)&=-\big(\tilde{\phi}_i(s)+\frac{\text{d}\Tilde{Q}_i}{\text{d}s}(s)\big),
    \end{align}
    and
    \begin{align}
        s\frac{\text{d}^{n+1}\tilde{\phi}_i}{\text{d}s^{n+1}}(s)&=\frac{\text{d}}{\text{d}s}\Big(s\frac{\text{d}^{n}\tilde{\phi}_i}{\text{d}s^{n}}(s)\Big)-\frac{\text{d}^{n}\tilde{\phi}_i}{\text{d}s^{n}}(s)\\ &=-\big((n+1)\frac{\text{d}^{n}\tilde{\phi}_i}{\text{d}s^{n}}(s)+\frac{\text{d}^{n+1}\Tilde{Q}_i}{\text{d}s^{n+1}}(s)\big),
    \end{align}
    using the induction hypothesis for the second equality.
    By Eq.~(\ref{eq:relation_phi_psi_moments}), we have:
    \begin{align}
        n\lim_{s\to0}\frac{\text{d}^{n-1}\tilde{\phi}_i}{\text{d}s^{n-1}}(s)=-\lim_{s\to0}\frac{\text{d}^{n}\Tilde{Q}_i}{\text{d}s^{n}}(s)\Rightarrow \lim_{s\to 0}s\frac{\text{d}^n\tilde{\phi}_i}{\text{d}s^n}(s)=0.
    \end{align}
    Therefore, we conclude that $\lim_{s\to 0}s\frac{\text{d}^n}{\text{d}s^n}\tilde{K}_{ij}(s)=0$.
\end{proof}

}

\new{
\section{Proofs for Section \ref{sec:weak_ergodicity_breaking}}\label{appendix_section_weak_ergodicity_breaking}

\begin{proof}{ of Proposition \ref{prop:generalized_Feynman_Kac}}
    For a fixed $s>0$, $\tilde{Q}_{ij}^{(\lambda)}(s)=e^{\lambda h_{ij}}\tilde{Q}_{ij}(s-\lambda g_i)$ and $\tilde{\phi}_{i}^{(\lambda)}(s)=\tilde{\phi}_{i}(s-\lambda g_i)$ are finite for all $\lambda$ in a sufficiently small neighborhood around $0$, and we can calculate $\tilde{G}_i(\lambda,s)$ by splitting the expectation value into a contribution before the first jump and after the first jump:
    \begin{align}
        G_i(t,\lambda)&=\mathbb{E}[e^{\lambda A(t)}~|~X_0=i]\\ &=\mathbb{P}(t<t_1~|~X_0=i)\mathbb{E}[e^{\lambda A(t)}~|~X_0=i, t<t_1]\\&+\mathbb{P}(t\geq t_1~|~X_0=i)\mathbb{E}[e^{\lambda A(t)}~|~X_0=i, t\geq t_1].
    \end{align}
    The first term is simply the survival probability up to time $t$ times the contribution to $A$ accumulated up to that time:
    \begin{align}
        \mathbb{P}(t<t_1~|~X_0=i)\mathbb{E}[e^{\lambda A(t)}~|~X_0=i, t<t_1]&=\phi_i(t)e^{\lambda g_it}\\ &=\phi_i^{(\lambda)}(t).
    \end{align}
    The second term describes the case that at least one transition away from the initial state $i$ occurs, so all possible transition times $t_1\leq t$ have to be considered, which contribute $e^{\lambda(g_it_1+h_{ij})}$ to $A$, depending on the next state $j\in V(\graph)$. Afterwards, the contribution is identical to a process starting in $j$ at time $t_1$, so we may express this term as:
    \begin{align}
        \mathbb{P}(t\geq t_1)\mathbb{E}[e^{\lambda A(t)}~|~X_0=i, t\geq t_1]&=\integ{[0,t]}{}{\sum_{j:(i,j)\in E(\graph)}e^{\lambda(g_it'+h_{ij})}\mathbb{E}[e^{\lambda A(t-t')}~|~X_0=j]~Q_{ij}(}{t')}\\ &=\sum_{j:(i,j)\in E(\graph)}\integ{[0,t]}{}{e^{\lambda(g_it'+h_{ij})}G_j(t-t',\lambda)Q_{ij}(}{t')}\\ &=\sum_{j:(i,j)\in E(\graph)}\integ{[0,t]}{}{G_j(t-t',\lambda)Q^{(\lambda)}_{ij}(}{t')}.
    \end{align}
    Combining these two summands and taking the Laplace transform of $G(t,\lambda)$ yields:
    \begin{align}
        \tilde{G}_i(\lambda,s)=\tilde{\phi}_i^{(\lambda)}(s)+\sum_{j:(i,j)\in E(\graph)}\tilde{Q}^{(\lambda)}_{ij}(s)\tilde{G}_j(\lambda,s).
    \end{align}
    Written in matrix notation, we see that:
    \begin{align}
         \big(\mathbb{1}-\tilde{Q}^{(\lambda)}(s)\big)\tilde{G}(\lambda,s)=\tilde{\phi}^{(\lambda)}(s).
     \end{align}
    Moreover, since $\sum_{j}\tilde{Q}_{ij}(s)<1$ if $i$ is not an absorbing state, also $\sum_{j}\tilde{Q}^{(\lambda)}_{ij}(s)<1$ for sufficiently small $|\lambda|$, so also $\tilde{G}(\lambda,s)$ is finite around $\lambda =0$.
\end{proof}

\begin{proof}{\textit{of Lemma \ref{lemma:memory_kernels_small_s}:}}
    Observe that:
    \begin{align}
        \frac{\text{d}\ln{\tilde{K}_{ij}(s)}}{\text{d}\ln{s}}&=\frac{s}{\tilde{K}_{ij}(s)}\frac{\text{d}}{\text{d}s}\tilde{K}_{ij}(s)\\&=\frac{s}{\tilde{K}_{ij}(s)}\Big(\frac{\tilde{K}_{ij}(s)}{s}+\frac{s\frac{\text{d}}{\text{d}s}\Tilde{Q}_{ij}(s)}{1-\Tilde{Q}_i(s)}+\frac{\frac{\text{d}}{\text{d}s}\Tilde{Q}_{i}(s)}{1-\Tilde{Q}_i(s)}\tilde{K}_{ij}(s)\Big)\\ &=1+\frac{s\frac{\text{d}}{\text{d}s}\Tilde{Q}_{ij}(s)}{\Tilde{Q}_{ij}(s)}+\frac{\frac{\text{d}}{\text{d}s}\Tilde{Q}_{i}(s)}{\tilde{\phi}_i(s)}.
    \end{align}
    Consider the second term in the limit $s\to 0$. The denominator satisfies $\tilde{Q}_{ij}(0)=p_{ij}>0$ and for the numerator, we can use that $xe^{-x}\leq e^{-1}$ for all $x\geq 0$ to see that for all $y\geq 0$, we have
    \begin{align}
        s\frac{\text{d}}{\text{d}s}\Tilde{Q}_{ij}(s)&=-s\integ{0}{\infty}{te^{-st}Q_{ij}(}{t)}\\ &\geq-s\integ{0}{y}{te^{-st}Q_{ij}(}{t)}-e^{-1}\integ{y}{\infty}{Q_{ij}(}{t)},
    \end{align}
    and also $s\frac{\text{d}}{\text{d}s}\Tilde{Q}_{ij}(s)\leq 0$.
    
    Since $\integ{0}{\infty}{Q_{ij}(}{t)}=p_{ij}<\infty$, the second integral can be made arbitrarily small by choosing $y$ accordingly and the first term vanishes as $s\to 0$. Hence, we see that
    \begin{align}
        \lim_{s\to 0}s\frac{\text{d}}{\text{d}s}\Tilde{Q}_{ij}(s)&=0 
        \Rightarrow \lim_{s\to 0}\frac{s\frac{\text{d}}{\text{d}s}\Tilde{Q}_{ij}(s)}{\Tilde{Q}_{ij}(s)}=0,
    \end{align}
    which implies that:
    \begin{align}
        \lim_{s\to 0}\big(\frac{\text{d}}{\text{d}\ln s}\ln\tilde{K}_{ij}(s)-\frac{\text{d}}{\text{d}\ln s}\ln\tilde{K}_{il}(s)\big)=0,
    \end{align}
    for any $i\in V(\graph)$ and $(i,j),(i,l)\in E(\graph)$. If the limit $\lim_{s\to 0}\frac{\text{d}}{\text{d}\ln s}\ln\tilde{K}_{ij}(s)$ exists, it is therefore independent of $j$ and given by:
    \begin{align}
        \lim_{s\to 0}\frac{\text{d}}{\text{d}\ln s}\ln\tilde{K}_{ij}(s)=1+\lim_{s\to 0}\frac{\frac{\text{d}}{\text{d}s}\Tilde{Q}_i(s)}{\tilde{\phi}_i(s)},
    \end{align}
    so we can set $\alpha_i:=-\lim_{s\to 0}\frac{\frac{\text{d}}{\text{d}s}\Tilde{Q}_i(s)}{\tilde{\phi}_i(s)}$. To see that $\alpha_i\in [0,1]$, note that the following holds for $s>0$:
    \begin{align}
        \frac{\text{d}}{\text{d}s}\Tilde{Q}_j(s)&=-\integ{0}{\infty}{te^{-st}Q_{j}(}{t)}<0\\ \tilde{\phi}_j(s)&=\frac{1-\Tilde{Q}_j(s)}{s}>0,
    \end{align}
    so $\frac{\frac{\text{d}}{\text{d}s}\Tilde{Q}_j(s)}{\tilde{\phi}_j(s)}<0$. Moreover,
    \begin{align}
        \frac{\frac{\text{d}}{\text{d}s}\Tilde{Q}_j(s)}{\tilde{\phi}_j(s)}&=-\frac{s\integ{0}{\infty}{te^{-st}Q_{j}(}{t)}}{1-\Tilde{Q}_j(s)}\\ &=-\frac{\integ{0}{\infty}{ste^{-st}Q_{j}(}{t)}}{\integ{0}{\infty}{(1-e^{-st})Q_{j}(}{t)}}\geq -1.
    \end{align}
    using the standard bound $ste^{-st}\leq 1-e^{-st}$ for $st\geq 0$.

    Finally, we show that if the limit exists for all $i\in V(\graph)$ and the network has a unique limiting state distribution $\pi$, then:
    \begin{align}
        \alpha_n>\alpha_m~\land~\pi_m>0 \Rightarrow \pi_n=0.
    \end{align}
    To that end, consider $\alpha_m=1-\lim_{s\to 0}\frac{\text{d}}{\text{d}\ln{s}}\ln\tilde{K}_{mi}(s)$ and $\alpha_n=1-\lim_{s\to 0}\frac{\text{d}}{\text{d}\ln s}\ln\tilde{K}_{nj}(s)$ and assume that $\alpha_n>\alpha_m$. Then:
    \begin{align}
        \lim_{s\to 0}\frac{\text{d}}{\text{d}\ln s}\ln\frac{\tilde{K}_{mi}(s)}{\tilde{K}_{nj}(s)}=\lim_{s\to 0}\frac{\text{d}}{\text{d}\ln s}\ln\tilde{K}_{mi}(s)-\lim_{s\to 0}\frac{\text{d}}{\text{d}\ln s}\ln\tilde{K}_{nj}(s)=\alpha_n-\alpha_m>0.
    \end{align}
    Hence, there is a $\delta >0$ such that $\forall s\leq\delta$:
    \begin{align}
        \frac{\alpha_n-\alpha_m}{2}&\leq\frac{\text{d}}{\text{d}\ln s}\ln\frac{\tilde{K}_{mi}(s)}{\tilde{K}_{nj}(s)}\leq 3\frac{\alpha_n-\alpha_m}{2} \\ \Rightarrow\frac{\alpha_n-\alpha_m}{2}\ln{\frac{\delta}{s}}&\leq\ln\frac{\tilde{K}_{mi}(\delta)}{\tilde{K}_{nj}(\delta)}+\ln\frac{\tilde{K}_{nj}(s)}{\tilde{K}_{mi}(s)}\leq 3\frac{\alpha_n-\alpha_m}{2}\ln{\frac{\delta}{s}}, 
    \end{align}
    by integration from $s$ to $\delta$. Letting $s\to 0$, it follows that:
    \begin{align}
        \lim_{s\to 0}\ln\frac{\tilde{K}_{nj}(s)}{\tilde{K}_{mi}(s)}=\infty\Leftrightarrow \lim_{s\to 0}\frac{\tilde{K}_{mi}(s)}{\tilde{K}_{nj}(s)}=0.
    \end{align}
    Also note that $\lim_{s\to0}\frac{\tilde{K}_{ij}(s)}{\tilde{K}_{il}(s)}=\frac{p_{ij}}{p_{il}}\in (0,\infty)$.
    Thus, for the ratio of the weights of two spanning trees $\mathcal{T}_m$ and $\mathcal{T}_n$ rooted at $m$ and $n$, respectively, we have:
    \begin{align}
        \lim_{s\to 0}\frac{w(\mathcal{T}_n)}{w(\mathcal{T}_m)}&=\lim_{s\to 0}\frac{\prod_{(a,b)\in E(\mathcal{T}_n)}\tilde{K}_{ab}(s)}{\prod_{(c,d)\in E(\mathcal{T}_m)}\tilde{K}_{cd}(s)}\\&=\Big(\lim_{s\to 0}\frac{\tilde{K}_{mi}(s)}{\tilde{K}_{nj}(s)}\Big)\frac{p_{nj}}{p_{mi}}\frac{\prod_{(a,b)\in E(\mathcal{T}_n)}p_{ab}}{\prod_{(c,d)\in E(\mathcal{T}_m)}p_{cd}}\\&=0,
    \end{align}
    where $(m,i)\in E(\mathcal{T}_n)$ and $(n,j)\in E(\mathcal{T}_m)$. Note that this identity follows from the fact that for $\tilde K_{ij}(s)$, the leading-order scaling for $s\to0$ depends only on the vertex $i$, and a tree  contains one outgoing edge from every vertex except its root. Using Eq.~(\ref{eq:transient_semi_MCTT}), this implies that:
    \begin{align}
        \frac{\pi_n}{\pi_m}=\lim_{s\to 0}\frac{\sum_{\mathcal{T}_{[n]}}w(\mathcal{T})}{\sum_{\mathcal{T}_{[m]}}w(\mathcal{T})}=0,
    \end{align}
    since $\lim_{s\to 0}\frac{s}{\tilde{K}_{ij}(s)}=\lim_{s\to 0}\frac{1-\tilde{Q}_i(s)}{\tilde{Q}_{ij}(s)}=0$ and because the number of spanning trees is finite in a finite graph.
    Since $m,n\in V(\graph)$ were arbitrary, $\pi_n,\pi_m>0$ implies $\alpha_n=\alpha_m$.
\end{proof}

}

\end{appendices}

\section*{Declarations}

\noindent\textbf{Conflict of interest.} 
The author has no competing interests to declare.

\noindent\textbf{Data availability}
The code and data used for the illustration in Fig.~\ref{fig:convergence} are available under \hyperlink{https://github.com/JVoits/Illustration-two-state-semi-Markov-network}{https://github.com/JVoits/Illustration-two-state-semi-Markov-network}.

\bibliography{sn-bibliography}

\end{document}